\documentclass[12pt]{article}
\usepackage[a4paper]{geometry}

\usepackage{amsfonts,amssymb,amsbsy,amsmath,amsthm,latexsym,natbib,graphicx}
\usepackage{epsf}
\usepackage[english]{babel}
\usepackage{algorithm}
\usepackage{subfigure,color,comment}

\usepackage{hyperref}

\usepackage{url}
\makeatletter
\g@addto@macro{\UrlBreaks}{\UrlOrds}
\makeatother

\newtheorem{theorem}{Theorem}
\newtheorem{proposition}{Proposition}

\newtheorem{corollary}{Corollary}
\theoremstyle{definition}

\newtheorem{remark}{Remark}

\newtheorem{assumptions}{Assumptions}

\newcommand{\blind}{0}

\newcommand{\phihat}{\widehat{\phi}}

\newcommand{\Dtil}{\widetilde{D}}

\newcommand{\pitil}{\tilde{\pi}}

\newcommand{\Prob}[1]{\mathbb{P} \left({#1}\right)}

\newcommand{\Expect}[1]{\mathbb{E} \left[{#1}\right]}
\newcommand{\Var}[1]{\mathsf{Var} \left[{#1}\right]}
\newcommand{\Cov}[1]{\mathsf{Cov} \left[{#1}\right]}

\newcommand{\Expects}[2]{\mathbb{E}_{{#1}} \left[{#2}\right]}
\newcommand{\Vars}[2]{\mathsf{Var}_{{#1}} \left[{#2}\right]}
\newcommand{\Covs}[2]{\mathsf{Cov}_{{#1}} \left[{#2}\right]}

\newcommand{\uv}{\underline{v}}

\newcommand{\md}{\mbox{d}}

\newcommand{\mbar}{\bar{m}}

\newcommand{\LeapFrog}{\mathsf{LeapFrog}}
\newcommand{\AAPS}{\mathsf{AAPS}}
\newcommand{\cS}{\mathcal{S}}
\newcommand{\cH}{\mathcal{H}}
\newcommand{\zcurr}{z^{\mathrm{curr}}}
\newcommand{\zprop}{z^{\mathrm{prop}}}

\newcommand{\xcurr}{x^{\mathrm{curr}}}
\newcommand{\xprop}{x^{\mathrm{prop}}}
\newcommand{\pcurr}{p^{\mathrm{curr}}}

\newcommand{\ESJD}{\mathsf{J}}
\newcommand{\Eff}{\mathsf{Eff}}
\newcommand{\Effbar}{\mathsf{Eff}^\mathrm{lb}}

\newcommand{\Khat}{\widehat{K}}

\graphicspath{{../Graphics/}}

\begin{document}

\def\spacingset#1{\renewcommand{\baselinestretch}%
{#1}\small\normalsize} \spacingset{1}

\if0\blind
{
  \title{\bf The Apogee to Apogee Path Sampler}
  \author{Chris Sherlock\thanks{All three authors gratefully acknowledge funding through EPSRC grant EP/P033075/1.} ,\hspace{.2cm}\\
    Department of Mathematics and Statistics, Lancaster University, UK\\
    Szymon Urbas, \hspace{.2cm}\\
    Department of Mathematics and Statistics, Lancaster University, UK\\
    and\\ 
    Matthew Ludkin, \hspace{.2cm}\\
  Darktrace, Cambridge, UK.}
  \maketitle
} \fi

\if1\blind
{
  \bigskip
  \bigskip
  \bigskip
  \begin{center}
    {\LARGE\bf Title}
\end{center}
  \medskip
} \fi

\bigskip

\begin{abstract}
  Amongst Markov chain Monte Carlo algorithms, Hamiltonian Monte Carlo (HMC) is often the algorithm of choice for complex, high-dimensional target distributions; however, its efficiency is notoriously sensitive to the choice of the integration-time tuning parameter. When integrating both forward and backward in time using the same leapfrog integration step as HMC, the set of \emph{apogees}, local maxima in the potential along a path, is the same whatever point (position and momentum) along the path is chosen to initialise the integration. We present the Apogee to Apogee Path Sampler (AAPS), which utilises this invariance to create a simple yet generic methodology for constructing a path, proposing a point from it and accepting or rejecting that proposal so as to target the intended distribution. We demonstrate empirically that AAPS has a similar efficiency to HMC but is much more robust to the setting of its equivalent tuning parameter, the number of apogees that the path crosses.
\end{abstract}

\noindent%
\vfill

\newpage
\spacingset{1.5} 

\section{Introduction}
\label{sec.intro}
Markov chain Monte Carlo (MCMC) is often the method of choice for estimating expectations with respect to complex, high-dimensional targets \cite[e.g.][]{GiRiSp1996,MCMChandbook}. Amongst MCMC algorithms, Hamiltonian Monte Carlo \cite[HMC, also known as Hybrid Monte Carlo;][]{DuKe1987} is known to offer a performance that scales better with the dimension of the state space than many of its rivals \cite[]{Neal2011,BePi2013}.

Given a target density $\pi(x)$, $x\in\mathbb{R}^d$, with respect to Lebesgue measure, and a current position, at each iteration HMC samples a momentum and numerically integrates Hamiltonian dynamics on a potential surface
\begin{equation}
U(x)=-\log \pi(x)
\end{equation}
to create a proposal that will either be accepted or rejected. 
As such, HMC has two main tuning parameters: the numerical integration step size, $\epsilon$, and the total integration time, $T$. Given $T$, guidelines for tuning $\epsilon$ have been available for some time \cite[]{BePi2013}; however, the integration time itself, is notoriously difficult to tune \cite[e.g.][]{Neal2011}, with algorithm efficiency often dropping sharply following only slight changes from the optimum $T$ value, and usually exhibiting approximately cyclic behaviour as $T$ increases.

The sensitivity is illustrated in the left panel of Figure \ref{fig.RosenSens}, in which HMC is applied to a modified Rosenbrock distribution (see Section \ref{sec.ToyTargets}) of dimension $d=40$. In the top half of this plot, the efficiency (see \eqref{eqn.gen.eff} for a precise definition), is given as a function of $\epsilon$ and the number of numerical integration steps, $L$. The bottom half of the panel shows the analogous plot for a modification of the HMC algorithm \cite[]{Mack1989,Neal2011}, which we refer to as \emph{blurred HMC} where at each iteration, the actual step-size is sampled (independently of all previous choices) uniformly from the interval $[0.8\epsilon,1.2\epsilon]$. This step was designed to mitigate the near reducibility of HMC that can occur when $T$ is some rational multiple of the integration time required to return close to the starting point, but as is visible from the plots, it also makes the performance of the algorithm more robust to the choice of $T$, and, we have found, often leads to a slightly more efficient algorithm. In both cases, the optimal tuning choice appears as a narrow ridge of roughly constant $T=L\epsilon$. Blurred HMC can be viewed as sampling the integration time $T$ uniformly from $[\frac{2}{3}T_*,T_*]$, where $T_*=1.2L\epsilon$. The approach of using a random integration time to introduce robustness has been extended recently to sampling uniformly from $[0,T_*]$ \cite[]{HoRa2021} and as an exponential variable with an expectation of $T_*$ \cite[]{rHMC}.

\begin{figure}
\begin{center}
    \includegraphics[scale=0.36]{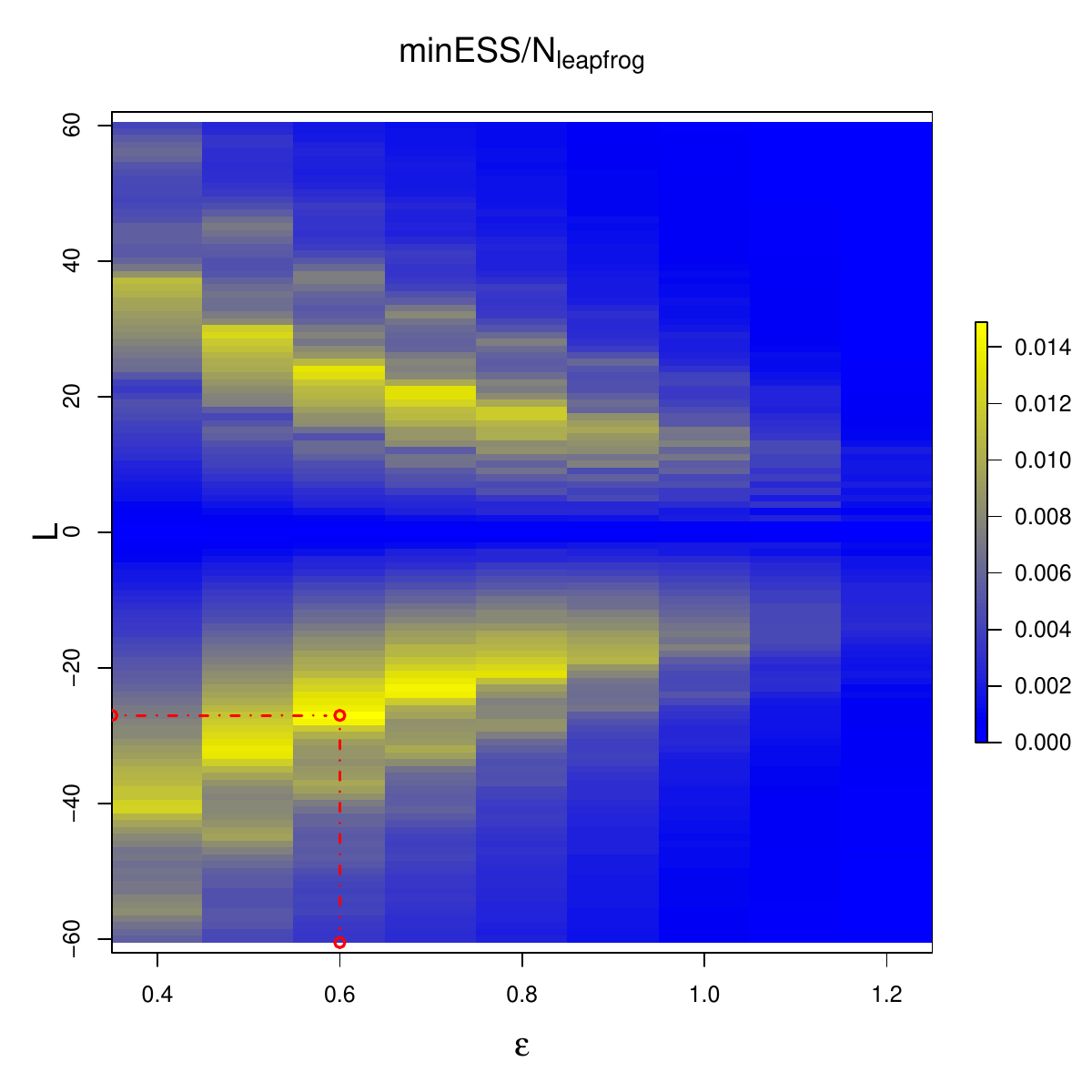}
    \includegraphics[scale=0.36]{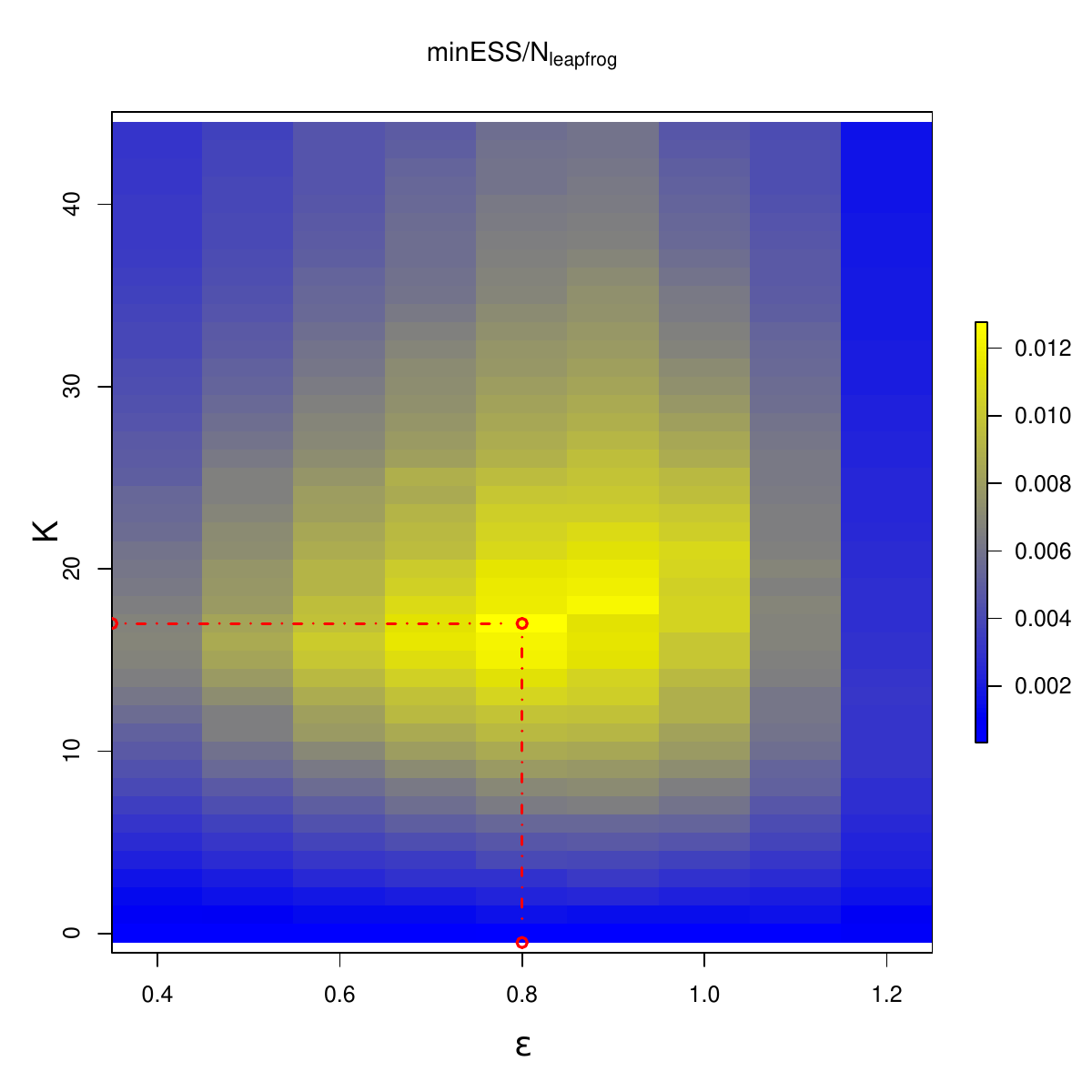}
    \end{center}
\caption{Efficiency, according to \eqref{eqn.gen.eff}, as a function of the tuning parameters for the 40-dimensional modified Rosenbrock target of Section \ref{sec.ToyTargets}. Left panel: HMC with positive $L$ values corresponding to the standard HMC algorithm and negative $L$ correspond to $|L|$ leapfrog steps of blurred HMC. Right panel: AAPS. Optimal parameter settings in red.
\label{fig.RosenSens}}
\end{figure}

Motivated by the difficulty of tuning $T$, \cite{HoGe2014} introduces the no U-turn sampler (NUTS). This uses the same numerical integration scheme as standard HMC, the leapfrog step, to integrate Hamiltonian dynamics both forward and backward from the current point, recursively doubling the size of the path until the distance between the points the farthest forward and backward in time stops increasing. Considerable care must be taken to ensure that the intended posterior is targeted, making the algorithm relatively complex; however, the no U-turn sampler is the default engine behind the popular STAN package \cite[]{Stan2020}, which has a relatively straightforward user interface.

Integration of Hamiltonian dynamics can be thought of as describing the position and momentum of a particle as it traverses the potential surface $U$. During its journey, provided $T$ is not too small, the particle will reach one or more local maximum, or \emph{apogee}, in the potential. The leapfrog scheme creates a path which is a discrete set of points rather than a continuum, and so (with probability $1$) apogees occur between consecutive points in the path; however, they are straightforward to detect. We call the set of points (positions and momenta) between two apogees a \emph{segment}. The discrete dynamics using the leapfrog step share several properties with the true dynamics, including the following: if we take the position and momentum of any point along the path, and integrate forward and backward for appropriate lengths of time we will create exactly the same path, and hence the same set of apogees and the same set of segments. This \emph{invariance} is vital to the correctness and flexibility of the algorithm presented in this article.

In Section \ref{sec.AAPS} we present the \emph{Apogee to Apogee Path Sampler} (AAPS). Like HMC, AAPS is straightforward to implement and has two tuning parameters, one of which is an integration step size, $\epsilon$. As with the no U-turn sampler, given a current point (position and momentum), AAPS uses the leapfrog step to integrate forwards and backwards in time. However, the integration stops when the path contains the segment in which the current point lies as well as $K$ additional segments, where $K$ is a user-defined tuning parameter. A point is then proposed from this set of $K+1$ segments and either accepted or rejected. The positioning of the current segment within the $K+1$ and the accept-reject probability are chosen precisely so that the algorithm targets the intended density. The invariance of the path to the starting position and momentum leads to considerable flexibility in the method for proposing a point from the set of segments, which in turn allows us to create an algorithm which enjoys a similar efficiency to HMC yet is extremely robust to the choice of $\epsilon$ and $K$. These properties are evident from the right-hand panel of Figure \ref{fig.RosenSens} and analogous plots for other distributions in Section \ref{sec.ToyTargets}. The robustness arises mainly from the proposal scheme; see the end of Section \ref{sec.chooseWeight}. The definition of $K$, however, also naturally caters to intrinsic properties of the target, such as its eccentricity, with  any externally imposed length (or time) scale largely irrelevant; the theoretical analysis of Section \ref{sec.GPlimit} makes this explicit in a simplified setting. 

Section \ref{sec.HamiltonAll} describes Hamiltonian dynamics and HMC. The AAPS is detailed in Section \ref{sec.AAPS}, and is compared empirically against HMC and the no U-turn sampler in Section \ref{sec.numExpts}. We conclude in Section \ref{sec.Discussion} with a discussion.

\section{Hamiltonian dynamics and Hamiltonian Monte Carlo}
\label{sec.HamiltonAll}
\subsection{Hamiltonian dynamics}
\label{sec.Hdynamics}
The position, $x$, and momentum, $p$, of a particle on a frictionless potential surface $U(x)$ evolve according to Hamilton's equations:

\begin{equation}
  \label{eq.HamiltonEqns}
  \frac{\md x}{\md t}=M^{-1}p
  ~~~\mbox{and}~~~
  \frac{\md p}{\md t}=-\nabla_x U.
\end{equation}

For a real object, $M$ is the mass of the object, a scalar, but the properties of the equations themselves that we will require hold more generally, when  $M$ is a symmetric, positive-definite mass matrix. The choice of $M$, whether for HMC or AAPS, is discussed at the end of Section \ref{sec.tuning}. We define $z_t=(x_t,p_t)$ and the map $\phi_t$ which integrates the dynamics forwards for a time $t$, so $z_t=\phi_t(z_0)$. The map, $\phi_t$ has the following fundamental properties \cite[e.g.][]{Neal2011}:
\begin{enumerate}
\item It is deterministic.
\item It has a Jacobian of $1$.
\item It is skew reversible: $\phi_t(x_t,-p_t)=(x_0,p_0)$.
\item It preserves the total energy $H(x,p)=U(x)+\frac{1}{2}p^\top M^{-1}p$.
\end{enumerate}

Except in a few special cases, the dynamics are intractable and must be integrated numerically, with a user-chosen time step which we will denote by $\epsilon$. The default method for Hamiltonian Monte Carlo, is the method which will be used throughout this article, the \emph{leapfrog} step; the leapfrog step itself is detailed in Appendix \ref{sec.leapfrog}. Throughout the main text of this article we denote the action of a single leapfrog step of length $\epsilon$ on a current state $z$ as $\LeapFrog(z;\epsilon)$. The leapfrog step satisfies Properties 1-3 above (see Appendix \ref{sec.leapfrog}), but it does not preserve the total energy.

Consider using $L$ leapfrog steps of size $\epsilon=t/L$ to approximately integrate the dynamics forward for time $t$. Since each individual leapfrog step satisfies Properties 1-3, so does the composition of $L$ such steps, which we denote $\phihat_t(z_0;\epsilon)$.

\subsection{Hamiltonian Monte Carlo}
\label{sec.HMC}
Hamiltonian Monte Carlo (HMC) creates a Markov chain which has a stationary distribution of $\pi(x)=\exp\{-U(x)\}$. 
Given a current position, $\xcurr$, and tuning parameters $\epsilon$ and $L$, a single iteration of the  algorithm proceeds as follows:
\begin{enumerate}
\item Sample a momentum $p_0\sim \mathsf{N}(0,M)$ and set $z_0=(\xcurr,p_0)$.
\item For $i$ in $1$ to $L$:
  \begin{itemize}
\item $z_i=\LeapFrog(z_{i-1};\epsilon)$.
  \end{itemize}
\item Let $\zprop=z_L$ and set $\alpha =1\wedge \pitil(\zprop)/\pitil(\zcurr)$.
  \item With probability $\alpha$, $\zcurr\gets \zprop$; else $\zcurr\gets \zcurr$.
\end{enumerate}

Here, with $\rho(p)$ denoting the density of the $\mathsf{N}(0,M)$ random variable,
\begin{equation}
  \label{eqn.defpitil}
\pitil(z)\equiv\pitil(x,p)=\pi(x)\rho(p)=\exp\{-H(x,p)\}.
\end{equation}
If $\xcurr$ is in its stationary distribution then $\pitil$ is the density of $z_0=(\xcurr,p_0)$.


\section{The Apogee to Apogee Path Sampler}
\label{sec.AAPS}

\subsection{Apogees and segments}

The left panel of Figure \ref{fig.showSegments} shows $L=50$ leapfrog steps of size $\epsilon=0.1$ from a current position, $x_0$ simulated randomly from a two-dimensional posterior $\pi(x)$ (with contours of $U(x)$ shown), and momentum $p_0$ simulated from a $\mathsf{N}(0,I_2)$ distribution. Different symbols and colours are used along the path, with both of these changing from step $l$ to step $l+1$ if and only if
\begin{equation}
  \label{eq.apogeeCondition}
  p_l^\top M^{-1}  \nabla U(x_l) >0
  ~~~\mbox{and}~~~
  p_{l+1}^\top M^{-1} \nabla U(x_{l+1}) <0.
\end{equation}
Intuitively, condition \eqref{eq.apogeeCondition} indicates when the ``particle'' has switched from moving ``uphill'' to moving ``downhill'' with respect to the potential surface $U(x)$. By \eqref{eq.HamiltonEqns}, $p^\top M^{-1} \nabla U\equiv \md x/\md t \cdot \nabla U\equiv \md U/ \md t$, so \eqref{eq.apogeeCondition} indicates a local maximum in $U(x_t)$ between $x_l$ and $x_{l+1}$. 

\begin{figure}
\begin{center}
    \includegraphics[scale=0.41]{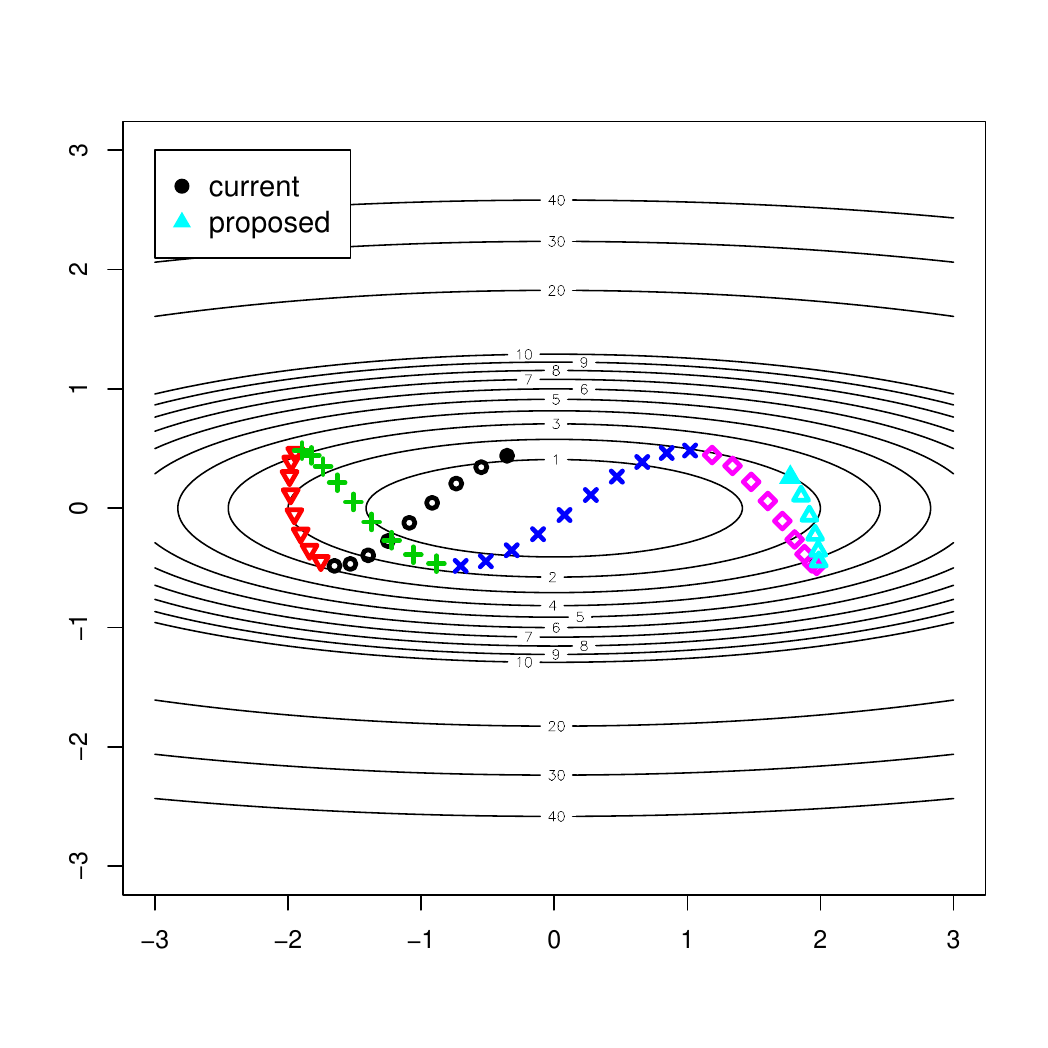}
    \includegraphics[scale=0.41]{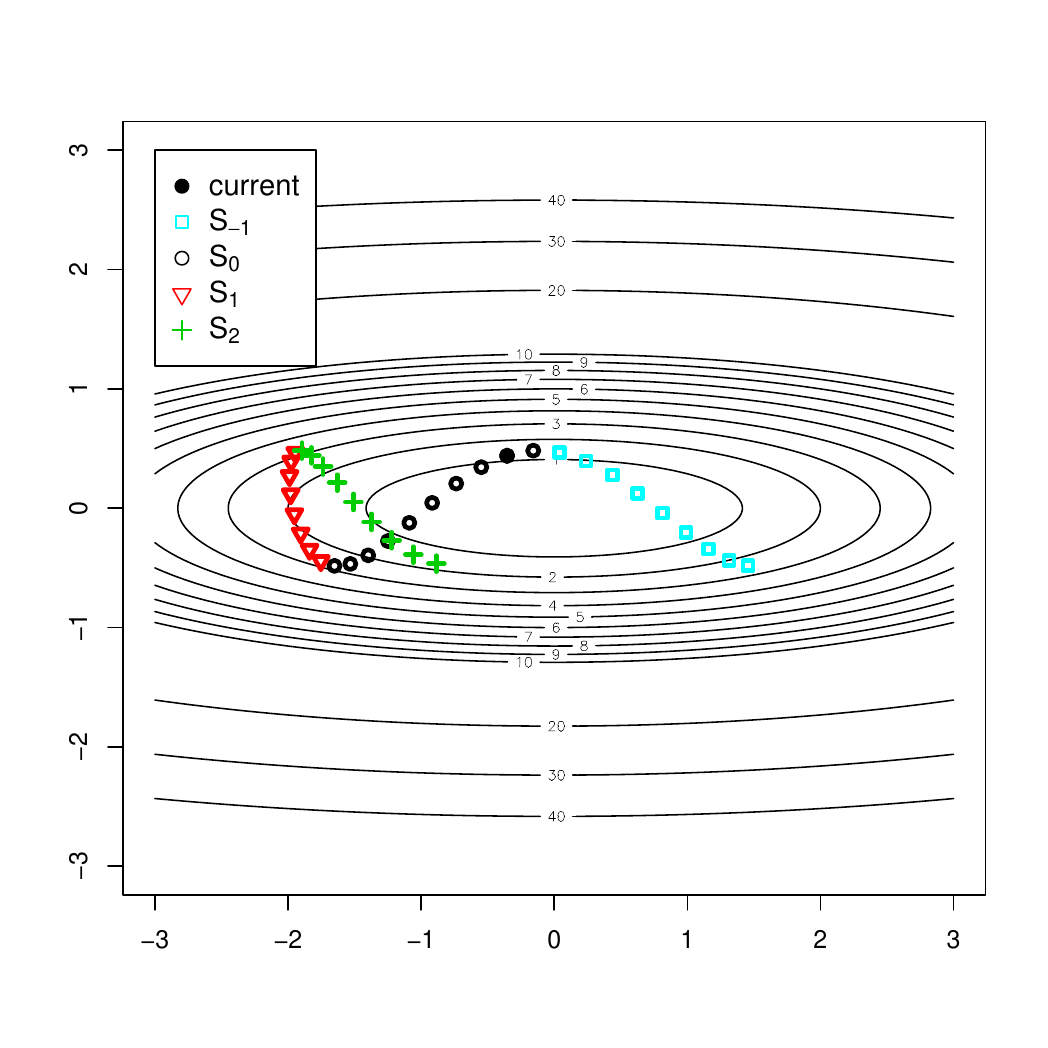}
    \end{center}
\caption{Left panel: $L=50$ leapfrog steps of size $\epsilon=0.1$ from the current point. Right panel: the current segment, $\cS_0$, and two segments forward and one segment backward using $\epsilon=0.1$. The current point, $x_0$ is simulated from a target, which has a density of $\pi(x)\propto \exp(-x_1^2/2-6x_2^2)$; $p_0$ is simulated from $\mathsf{N}(0,I_2)$. Different colours and symbols are used for each segment along the path.
\label{fig.showSegments}}
\end{figure}

Between such a pair of points at $l$ and $l+1$ there is a hypothetical point, $a$ with $l<a<l+1$ where the particle's potential has reached a local maximum, which we call an \emph{apogee}. Under the exact, continuous dynamics, this point would, of course, be realised, but under the discretised dynamics the probability of this is $0$.
We call each of the realised sections of the path between a pair of consecutive apogees (\emph{i.e.}, each portion with a different colour and symbol in Figure \ref{fig.showSegments}) a \emph{segment}. Each segment consists of the time-ordered list of position and momentum at each point between two consecutive apogees.

Instead of integrating forward for $L$ steps, one can imagine integrating both forwards and backwards in time from $z_0=(x_0,p_0)$ until a certain number of apogees have been found. The right pane of Figure \ref{fig.showSegments} shows the segment to which the current point belongs, which we denote $\cS_0(z_0)$, together with two segments forward and one segment backward. We denote the $j$th segment forward by $\cS_j(z_0)$ and the $j$th segment backward by $\cS_{-j}(z_0)$. We abbreviate the ordered  collection of segments from $\cS_a(z_0)$ to $\cS_b(z_0)$ as $\cS_{a:b}(z_0)$. Thus, the right panel of Figure \ref{fig.showSegments} shows the positions from $\cS_{-1:2}(z_0)$. For a particular point $z'=(x',p')\in \cS_{a:b}$, we denote the segment to which it belongs by $\cS_{\#}(z';z_0)$. 

The following \emph{segment invariance property} is vital to both the simplicity and correctness of our algorithm. For $a\le b$, any $z$ and any $z'\in \cS_{a:b}(z)$
\begin{equation}
  \label{eq.seg.invariance}
  \cS_{a':b'}(z')\equiv \cS_{a:b}(z)
  ,~~~\mbox{where}~~~a'=-c',~b'=K-c'~\mbox{and}~c'=c+\cS_{\#}(z';z).
  \end{equation}

The quantities $a'$, $b'$ and $c'$ correspond to $a$, $b$ and $c$ but from the point of view of $z'$ rather than $z$. 
For the right panel of Figure \ref{fig.showSegments}, for example, picking any $z'=(x',p')$ from $\cS_1(z_0)$, $\cS_{-2:1}(z')$ would give the same ordered set of segments as illustrated in the figure. This is because the numerical integration scheme is deterministic and skew reversible, so the apogees would all occur in exactly the same positions with exactly the same (up to a possible sign flip) momenta.

\subsection{The AAPS algorithm}

We now introduce our algorithm, the Apogee to Apogee Path Sampler, $\AAPS$. The algorithm requires a weight function $w:\mathbb{R}^{4d}\rightarrow [0,\infty)$, where $d$ is the dimension of the target. Weight functions are investigated in more detail in Sections \ref{sec.chooseWeight}, but for now it might be helpful keep in mind the simplest that we consider: $w(z,z')=\pitil(z')$.

  Given a step-size $\epsilon$, a non-negative integer, $K$, a mass matrix, $M$ and a current position $\xcurr$, one iteration of $\AAPS$ proceeds as follows:
  \begin{enumerate}
  \item Sample a momentum $p_0\sim \mathsf{N}(0,M)$ and set $\zcurr=z_0=(\xcurr,p_0)$.
    \item Simulate $c$ uniformly from $\{0,1,\dots,K\}$; set $a=-c$ and $b=K-c$. 
  \item Create $\cS_{a:b}$ by leapfrog stepping \underline{forward} from $(x_0,p_0)$ and then \underline{backward} from $(x_0,p_0)$.
  \item Propose $\zprop$  w.p. $\propto w(\zcurr,\zprop)~1(\zprop \in \cS_{a:b}(\zcurr))$.
  \item With a probability of 
    \begin{equation}
      \label{eq.AAPSalpha}
    \alpha=1\wedge
    \frac{\pitil(\zprop)w(\zprop,\zcurr)\sum_{z\in\cS_{a:b}(\zcurr)}w(\zcurr,z)}
      {\pitil(\zcurr)w(\zcurr,\zprop)\sum_{z\in\cS_{a:b}(\zcurr)}w(\zprop,z)}
      \end{equation}
      set $\zcurr\gets \zprop$ else $\zcurr \gets \zcurr$.
  \item Discard $\pcurr$ and retain $\xcurr$.
  \end{enumerate}

The Metropolis-Hastings formula \eqref{eq.AAPSalpha} arises because, out of the allowable proposals once $c$ has been chosen, the probability of proposing $\zprop$ is $q(\zprop|\zcurr)=w(\zcurr,\zprop)/\sum_{z\in \cS_{a:b}} w(\zcurr,z)$. 

  \begin{proposition}
The $\AAPS$ algorithm satisfies detailed balance with respect to the extended posterior $\pitil(x,p)$.
  \end{proposition}
  \begin{proof}
    Step 1 preserves $\pitil$ because $p$ is independent of $x$ and is sampled from its marginal. 
    
    It will be helpful to define the system from the point of view of starting at $\zprop$. Let $a'$, $b'$ and $c'$ be as defined in \eqref{eq.seg.invariance}, but with $z'=\zprop$. Then, since $\cS_{\#}(\zprop;\zcurr)=-\cS_{\#}(\zcurr;\zprop)$,
    \[
-c \le 0 \le K-c, ~-c \le \cS_{\#}(\zprop;\zcurr) \le K-c \Leftrightarrow -c'\le \cS_{\#}(\zcurr;\zprop) \le K-c',~ -c'\le 0 \le K-c'.
\]
Equivalently, $
    1(c\in \{0,\dots,K\})1(\zprop\in \cS_{a:b}(\zcurr))
    \equiv
    1(c'\in \{0,\dots,K\})1(\zcurr\in \cS_{a':b'}(\zprop))$.    
    Moreover, segment invariance \eqref{eq.seg.invariance} is equivalent to
$z\in \cS_{a:b}(\zcurr)\iff z\in \cS_{a':b'}(\zprop)$.

The resulting chain satisfies detailed balance with respect to $\pitil$ because
\[
\pitil(\zcurr)\times\Prob{c}\times\Prob{\mbox{propose }\zprop|c}\times\Prob{\mbox{accept }\zprop|\mbox{proposed},c}
\]
is
\begin{align*}
\pitil(\zcurr)\times\frac{1(c\in \{0,\dots,K\})}{K+1}
&\times\frac{w(\zcurr,\zprop)1(\zprop\in \cS_{a:b}(\zcurr))}{\sum_{z\in \cS_{a:b}(\zcurr)} w(\zcurr,z)}\\
&\times
1\wedge
    \frac{\pitil(\zprop)w(\zprop,\zcurr)\sum_{z\in\cS_{a:b}(\zcurr)}w(\zcurr,z)}
         {\pitil(\zcurr)w(\zcurr,\zprop)\sum_{z\in\cS_{a:b}(\zprop)}w(\zprop,z)},
         \end{align*}
which is
\[
\frac{1(c\in \{0,\dots,K\})1(\zprop\in \cS_{a:b}(\zcurr))}{K+1}
\times
      \left\{
\frac{\pitil(\zcurr)w(\zcurr,\zprop)}{\sum_{z\in\cS_{a:b}(\zcurr)}w(\zcurr,z)}\wedge
\frac{\pitil(\zprop)w(\zprop,\zcurr)}{\sum_{z\in\cS_{a':b'}(\zprop)}w(\zprop,z)}
\right\}.
\]
This expression is invariant to $(\zcurr,a,b,c) \leftrightarrow (\zprop,a',b',c')$.
\end{proof}

  \begin{remark} AAPS algorithm could be applied with a numerical integration scheme that does not have a Jacobian of $1$. In such a scheme, however, the Jacobian would need to be included in the acceptance probability; moreover, a Jacobian other than $1$ would lead to greater variability in $\pitil$ along a path and, we conjecture, to a reduced efficiency. 
    \end{remark}

The path $\cS_{a:b}$ may visit parts of the statespace where the numerical solution to \eqref{eq.HamiltonEqns} is unstable and the error in the Hamiltonian may increase without bound, leading to wasted computational effort as large chunks of the path may be very unlikely to be proposed and accepted. Indeed it is even possible for the error to increase beyond machine precision. The no U-turn sampler suffers from a similar problem and we introduce a similar stability condition to that in \cite{HoGe2014}. We require that 
\begin{equation}
  \label{eq.Delta}
\max_{z\in \cS_{a:b}(\zcurr)} H(z)-\min_{z\in \cS_{a:b}(\zcurr)} H(z) <\Delta,
\end{equation}
for some prespecified parameter $\Delta$. This criterion can be monitored as the forward and backward integrations proceed, and if at any point the criterion is breached, the path is thrown away: the proposal is automatically rejected. Segment invariance means that $z\in\cS_{a:b}(\zcurr)\iff z\in \cS_{a':b'}(\zprop)$, so the same rejection would have occured if we created the path from any proposal in $\cS_{a:b}(\zcurr)$, and detailed balance is still satisfied. For the experiments in Section \ref{sec.numExpts} we found that a value of $\Delta=1000$ was sufficiently large that the criterion only interfered when something was seriously wrong with the integration. Step 3 of the algorithm then becomes:

\begin{enumerate}
  \setcounter{enumi}{2}
\item Create $\cS_{a:b}$ by leapfrog stepping \underline{forward} from $(x_0,p_0)$ and then \underline{backward} from $(x_0,p_0)$. If condition \eqref{eq.Delta} fails then go to 6.
\end{enumerate}

For a $d$-dimensional Markov chain, $\{X_t\}_{t=0}^\infty$, with a stationary distribution of $\pi$, the asymptotic variance is $V_f:=\lim_{n\uparrow \infty}n\Var{\frac{1}{n}\sum_{i=1}^n f(X_i)}$. Thus, after burn in, $\Var{\frac{1}{n}\sum_{i=1}^n f(X_i)} \approx V_f/n$. The effective sample size is the number of iid samples from $\pi$ that would lead to the same variance: $\mathrm{ESS}_f=n \Vars{\pi}{f}/V_f$. Let $\mathrm{ESS}_i$ be an empirical estimate of the ESS for the $i$th component of $X$ (\emph{i.e.}, $f$ gives the $ith$ component) and let $n_{leap}$ be the total number of leapfrog steps taken during the run of the algorithm. We measure the efficiency of an algorithm as:
\begin{equation}
  \label{eqn.gen.eff}
  \mathrm{Eff}=\frac{1}{n_{leap}}\min_{i=1,\dots,d}\mathrm{ESS}_i.
\end{equation}
Since the leapfrog step is by far the most computationally intensive part of the algorithm, $\mathrm{Eff}$ is proportional to the number of effective iid samples generated per second in the worst mixing component. 

  \subsection{Choice of weight function}
\label{sec.chooseWeight}

The weight function $w(z,z')=\pitil(z')$ mentioned previously is a natural choice, and substitution into \eqref{eq.AAPSalpha} shows that this leads to an acceptance probability of $1$; however, it turns out not to be the most efficient choice. For example, intuitively the algorithm might be more efficient if points on the path that are further away from the current point are more likely to be proposed. To investigate the effects of the choice of $w$ we examine six possibilities:

  \begin{enumerate}
  \item $w(z,z')=\pitil(z')$, which leads to $\alpha=1$. 
  \item $w(z,z')=||x'-x||^2$, squared jumping distance (SJD).
  \item $w(z,z')=\pitil(z')||x'-x||^2$, SJD modulated by target.
  \item $w(z,z')=||x'-x||$, absolute jumping distance (AJD).
  \item $w(z,z')=\pitil(z')||x'-x||$, AJD modulated by target. 
    \item $w(z,z')=\pitil(z')~1(z'\in \cH(z))$, where $\cH(z)$ is described below; this also leads to $\alpha=1$.
    \end{enumerate}

 Scheme 6 essentially partitions $\cS_{a:b}(\zcurr)$ into two halves of roughly equal total $\pitil$ and then proposes only values from the half that does not contain $\zcurr$, $\cH(\zcurr)$; details are given in Appendix \ref{sec.WSsix}.

\begin{figure}
\begin{center}
  \includegraphics[scale=0.39]{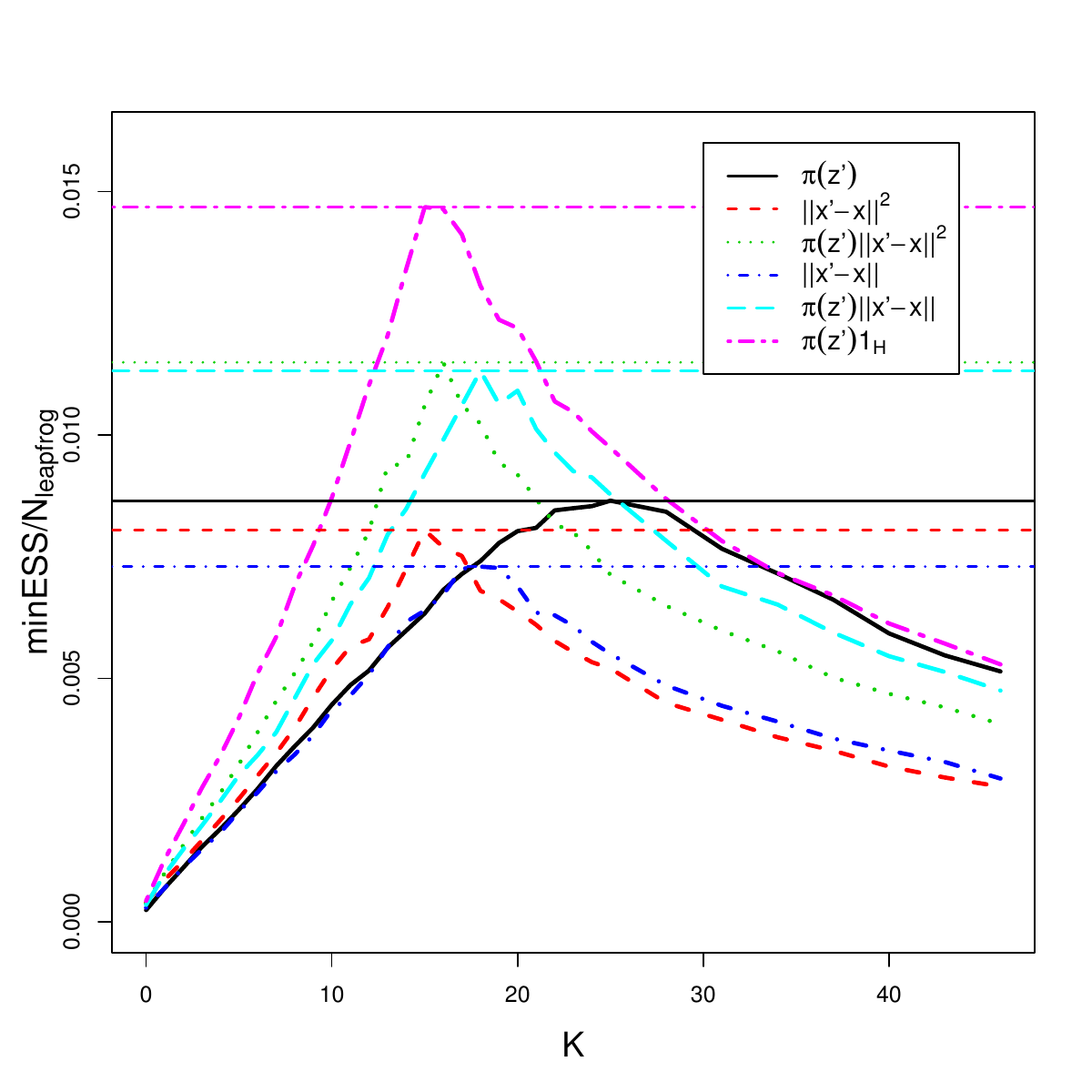}
      \includegraphics[scale=0.39]{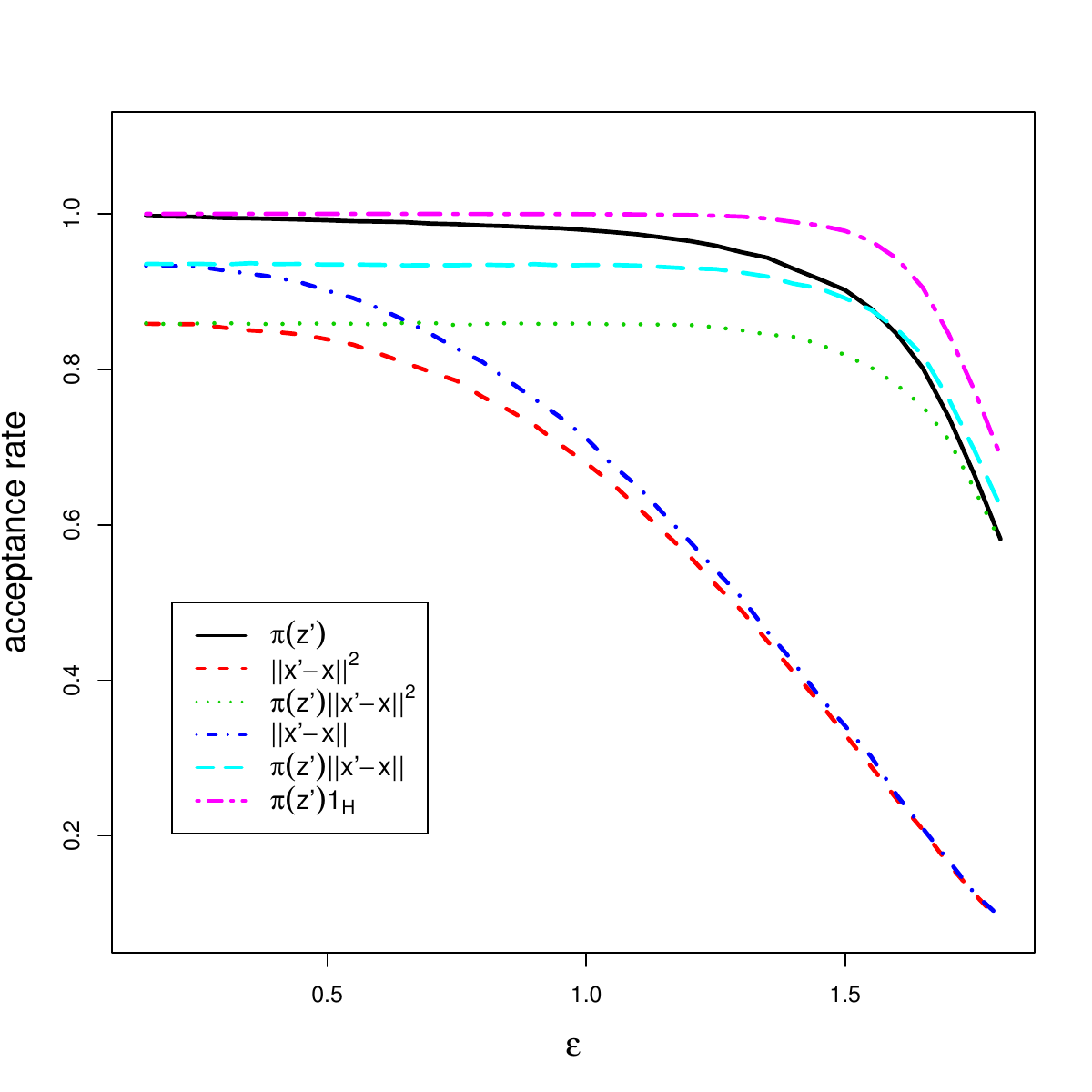}
    \end{center}
\caption{Left: efficiency (see \eqref{eqn.gen.eff}) of $\AAPS$ as a function of $K$ when $\epsilon=1.2$. Right: acceptance rate as a function of $\epsilon$ when $K=15$. There is one curve for each of the six choices of weight function. The single horizontal line associated with each curve  in the left panel indicates the maximum efficiency achieved. The target is $\pi_G^H(x)$ (see Section \ref{sec.ToyTargets}) with $d=40$ and $\xi=20$.
\label{fig.compareWeights}}
\end{figure}

The left panel of Figure \ref{fig.compareWeights} shows, for a particular choice of $\epsilon$ and target, the efficiency of $\AAPS$ as a function of $K$ for each of the six weight schemes. Scheme 6 is the most efficient; however Schemes 3 and 5 each of which involve some measure of jumping distance modulated by $\pitil$ are not far behind. Indeed there is only a factor of two between the least and most efficient. Very similar relative performances were found for all the other toy targets from Section \ref{sec.ToyTargets} and across a variety of choices of $\epsilon$, except that when $\epsilon$ becomes small, modulation of SJD or AJD by $\pitil$ makes little difference since relative changes in $\pitil$ are small.

A simple heuristic can explain the difference of a factor of nearly $2$ between  Scheme 1 and Scheme 6 in the case where $\pitil$ is approximately constant; for example, when $\epsilon$ is small. $\cS_{a:b}$ is constructed prior to making the proposal, so the computational effort does not depend on the scheme; thus efficiency can be measured crudely in terms of the squared jumping distance of the proposal \cite[e.g.][]{RobRos2001,SheRob2009,BePi2013}, since it is accepted with probability $1$. Without loss of generality, we rescale $\cS_{a:b}$ to have unit length. For two independent $\mathsf{Unif}(0,1)$ variables $U_1$ and $U_2$, Scheme 1 is equivalent to an expected squared jumping distance (ESDJ) of $\Expect{(U_1-U_2)^2}=1/6$, whereas Scheme 2 is equivalent to an ESJD of $\Expect{\{(1-U_1/2)-U_2/2)\}^2}= 7/24$; the ratio between the two ESJDs is $7/4$.

A naive implementation of each of the above weighting schemes would require storing each of the points in $\cS_{a:b}(\zcurr)$, which has an associated memory cost of $\mathcal{O}(K)$ and an exact cost that varies from one iteration to the next as the number of points in each segment is not fixed. For all schemes except the last there is a simple mechanism for sampling $\zprop$ with a fixed $\mathcal{O}(1)$ memory cost; however, this would be useless if calculation of the acceptance ratio $\alpha$  in \eqref{eq.AAPSalpha} still required storage of $\mathcal{O}(K)$. For Schemes 1, 2 and 3, however, it is also possible to calculate $\alpha$  with a fixed $\mathcal{O}(1)$ memory usage. We have found that this has a negligible effect on the CPU cost but the impact on the peak memory footprint of the code is substantial, decreasing by a factor of around $17$ in $d=40$, $27$ in $d=100$ and $42$ in $d=800$. Since there is little otherwise to choose between the schemes, we opt for the most efficient of these, Scheme 3, $w(z,z')=||x'-x||^2\pitil(x')$ and \emph{apply this  thoughout the remainder of this article}. Appendix \ref{sec.memOone} details the implementation of Schemes 1-3 with $\mathcal{O}(1)$ memory cost.

\subsection{Robustness and efficiency}
The robustness of AAPS with Scheme 3 which is visible in Figure \ref{fig.RosenSens} and similar figures in Section \ref{sec.ToyTargets}, arises from two aspects of the algorithm. We first consider robustness to the choice of $\epsilon$, then robustness to the choice of $K$.

For Scheme 2, as with HMC itself, the acceptance probability contains the ratio $\pitil(\zprop)/\pitil(\zcurr)$. The error in the total energy of HMC is proportional to $\epsilon^2$, so as $\epsilon$ increases the acceptance rate drops substantially when the error becomes $\mathcal{O}(1)$. By contrast, the acceptance probability for Scheme 3 contains only weighted sums of the $\pitil$ in both the numerator and denominator so the empirical acceptance rate is much more robust to increasing $\epsilon$. This effect can be seen in the right panel of Figure \ref{fig.compareWeights}.

Robustness to the choice of $K$ arises because AAPS samples a point from the whole set of segments rather than choosing the end point of the integration as HMC does. Figure \ref{fig.sensDiffSchemes} (bottom left) of Appendix \ref{sec.moreRobustnessFigures} echoes Figure \ref{fig.RosenSens}, but for an alternative version of AAPS which always sets $c=0$ and always samples the proposal from segment $K$ via weighting Scheme 3. At its optimum, the algorithm is slightly more efficient than the version of AAPS that we recommend; however, the efficiency drops much more sharply when $K$ is larger or smaller than its optimal value.

Both of the above robustness properties are shared by Scheme 1 (also demonstrated in Figure \ref{fig.sensDiffSchemes}); however, Scheme 3 is more efficient than Scheme 1 precisely because it preferentially targets proposals that are further from the current value.


\subsection{Tuning}
\label{sec.tuning}
Since AAPS with weighting Scheme 3 is much more robust to the choice of either tuning parameter than HMC is, precise tuning is less important than it is for HMC. Thus, we present only brief guidelines here; further heuristics and empirical evidence for them are presented in Appendices \ref{sec.TuningEpsilon} and \ref{sec.tuning.K}.

For any given $K>0$, we recommend increasing $\epsilon$ from some small value until the empirical acceptance starts to change, stopping when the change from the small-$\epsilon$ acceptance rate is no more than $3\%$. Empirically, we have found that such minor changes in acceptance rate correspond to substantial changes in the total energy over the $K+1$ segments, such that the modulation of SJD by $\pitil$ starts to have a negative impact on the choice of $\zprop$.

For a given value of $\epsilon$ we recommend a short tuning run of AAPS using a large value, $K_*$, of $K$ and then choosing a sensible $K\in \{0,\dots,K_*\}$ according to the most popularly proposed segment number using a diagnostic that we describe in Appendix \ref{sec.tuning.K}.

As with HMC and the no U-turn sampler, the mass matrix, $M$, used by AAPS can also be tuned, and mixing will be optimised if $M^{-1}\approx \Vars{\pi}{X}$, as approximated from a tuning run. However, the matrix-vector multiplication required for simulating $p_0$ and in each leapfrog step can be expensive, so it is usual to choose a diagonal $M$ instead.

\subsection{Gaussian process limit}
\label{sec.GPlimit}
Intuitively, the more eccentric a target the more apogees one might expect per unit time. This section culminates in an expression which makes this relationship concrete for a product target where each component has its own length scale. 
With the initial condition $(X_0,P_0)$ sampled from $\pitil$, $P_t^\top M^{-1} \nabla U(X_t)$ can be considered as a random function of time.
We first show that when the true Hamiltonian dynamics are used, subject to conditions, a rescaling of this dot-product tends to a stationary Gaussian process as the number of components in the product tends to infinity. A formula for the expected number of apogees per unit time follows as a corollary.

We consider a $d$-dimensional product target with a potential of
\begin{equation}
  \label{eq.product.target}
U^{(d)}(x)=\sum_{i=1}^d g\left(\sqrt{\nu^{(d)}_i} x^{(d)}_i\right)+\mbox{constant}
\end{equation}
for some $g:\mathbb{R}\to\mathbb{R}$ and values $\nu^{(d)}_i>0$, $i=1,\dots,d$, and where $\int \exp\{-U^{(d)}(x)\}\md x=1$.

HMC using a diagonal mass matrix, $M$, and a product target is equivalent to HMC using an identity mass matrix and a target of $M^{-1/2} x$ \cite[e.g.][]{Neal2011}, which, in the case of \eqref{eq.product.target} is also a product target, but with different $\nu_i$. For simplicity, therefore, we assume the identity mass matrix throughout this section. We also consider the true Hamiltonian dynamics which are approached in the limit as $\epsilon\downarrow 0$, but approximate the dynamics for small to moderate $\epsilon$ reasonably well.

Define the scaled dot product at time $t$ given an initial position of $x^{(d)}_0$ and momentum of $p^{(d)}_0$ as
\begin{equation}
  \label{eq.scaled.dot.product}
  D^{(d)}(t;x^{(d)}_0,p^{(d)}_0):=
  \frac{1}{\sqrt{d}}~p^{(d)}\left(t\right)\cdot \nabla_x U|_{x^{(d)}(t)}
  =
  \frac{1}{\sqrt{d}}~\sum_{i=1}^d \sqrt{\nu^{(d)}_i} g'\left(\sqrt{\nu_i^{(d)}} x^{(d)}_i\right) p^{(d)}_i\left(t\right).
\end{equation}

We define the one-dimensional densities with respect to Lebesgue measure
\[
\pi_\nu(x)=\sqrt{\nu}\exp\{-g(\sqrt{\nu}x)\}
~~~\mbox{and}~~~
\rho_1(p)=\frac{1}{\sqrt{2\pi}}\exp\{-p^2/2\},
\]
and let $\pi_\nu$ and $\rho_1$ denote the corresponding measures. The joint density and measure are $\pitil_\nu(x,p;\nu)=\pi_\nu(x;\nu)\rho_1(p)$ and $\pitil_\nu$. 

\begin{assumptions}
  \label{ass.g}
We assume that $g\in C^1$, that there is a $\delta>0$ such that
\begin{equation}
  \label{eq.finite.moment.one}
\Expects{X\sim \pi_1}{|g'(X)|^{2+\delta}}<\infty,
\end{equation}
and that for each $y_0\in \mathbb{R}$, there is a unique, non-explosive solution $a(t;y_0,y'_0)$ to the initial value problem:
\begin{equation}
  \label{eq.ODE}
  \frac{\md^2 y}{\md t^2}=-g'(y);~y(0)=y_0,~y'(0)=y'_0.
\end{equation}
\end{assumptions}

Theorem \ref{thrm.GP} is proved in Appendix \ref{sec.prove.GP.thrm}.

\begin{theorem}
  \label{thrm.GP}
  Let the potential be defined as in \eqref{eq.product.target} and where $g$ satisfies the assumptions around \eqref{eq.finite.moment.one} and \eqref{eq.ODE}. 
  Further, let $\mu$ be a distribution with support on $\mathbb{R}^+$  with
  \begin{equation}
    \label{eq.finite.moment.two}
    \Expects{\nu\sim\mu}{\nu^{1+\delta/2}}<\infty
  \end{equation}
  for some $\delta>0$, and let $\nu_i\stackrel{iid}{\sim}\mu$. Define
  \begin{equation}
    \label{eq.SGPcov}
V(t)=  \Expect{\nu~g'(X) P ~g'(\sqrt{\nu} X_t(X,P))~a'(\sqrt{\nu}t;X,P)},
\end{equation}
where the expectation is over the independent variables $X\sim \pi_1$, $P\sim \rho_1$ and $\nu\sim \mu$, and assume that for any finite sequence of $n$ distinct times $(t_1,\dots,t_n)$, the $n\times n$ matrix $\Sigma$ with $\Sigma_{i,j}=V(t_j-t_i)$ is positive definite.

Let $D^{(d)}$ be the scaled dot product defined in \eqref{eq.scaled.dot.product},
 and let $(X_0,P_0)\sim \pitil$; then
 \begin{equation}
   \label{eqn.SGPfull}
D^{(d)}\Rightarrow \Dtil \sim \mathsf{SGP}\left(0,V\right),
 \end{equation}
 as $d\rightarrow \infty$,
where $Y\sim\mathsf{SGP}(b,V)$ denotes that $Y$ is a one-dimensional stationary Gaussian process with an expectation of $b$ and a covariance function of $\Cov{Y_t,Y_{t'}}=V(t'-t)$.
\end{theorem}

\cite{Ylvisaker1965} shows that the expected number of zeros over a unit interval of a stationary Gaussian process with a unit variance is:
$\frac{1}{\pi}\sqrt{-C''(0)}$, where $C$ is its covariance function. Zeroes can be either apogees or perigees (local minima in $U(x)$ along the path), and these alternate. Hence, with some work (see Appendix \ref{sec.prove.GP.corollary}), we obtain the following:

\begin{corollary}
\label{cor.EN}
  The expected number of apogees of $\Dtil$ over a time $T$ is
\[
N(T)\propto
T\sqrt{\Expect{\nu}}\times \sqrt{\frac{\Expect{\nu^2}}{\Expect{\nu}^2}}.
\]
\end{corollary}
Since $\nu$ is a squared inverse scale parameter, the first term in the product simply relates the time interval to the overall length scale of the target ($1/\sqrt{\Expect{\nu}}$), given that $P_0\sim N(0,I_d)$ and an identity mass matrix is used. The second part of the product is more interesting and shows how the expected number of apogees per unit time interval increases with variability in the squared inverse length scales of the components of the target. The second term also makes it plain that the tuning parameter $K$ relates to properties intrinsic to the target; unlike $L$, its impact is relatively unaffected by a uniform redefinition of the length scale of the target or by the choice of $\epsilon$, the latter providing yet another reason for the robustness of $\AAPS$.

\section{Numerical Experiments}
\label{sec.numExpts}

In this section we compare, AAPS with HMC, blurred HMC (see Section \ref{sec.intro}) and the no U-turn sampler over a variety of targets. For fairness we use the basic no U-turn sampler from \cite{HoGe2014}, without it needing to adaptatively tune $\epsilon$; instead tuning $\epsilon$ using a fine grid of values. For both varieties of HMC we use a grid of $\epsilon$ and $L$ values, and for AAPS we use a grid of $\epsilon$ and $K$ values; in each case we choose the combination that leads to the optimal efficiency. To ensure that the various toy targets are close to representing the difficulties encountered with targets from complex statistical models, all algorithms use the identity mass matrix.

All code ($\AAPS$/HMC/no U-turn sampler) was written in \texttt{C++}; the effective sample size was estimated using the \texttt{R} package \texttt{coda}, and the mean number of leapfrog steps per iteration (for AAPS and the no U-turn sampler) was output from each run as a diagnostic. 

\subsection{Toy targets}
\label{sec.ToyTargets}
Here we investigate performance across a variety of targets with a relatively simple functional form, and across a variety of dimensions. Many are product densities with independent components: Gaussian, logistic and skew-Gaussian. We consider different, relatively large ratios $\xi$ between the largest and smallest length scales of the components in a target, as well as different sequences of scales from the smallest to the largest. We also consider a modification of the Rosenbrock banana-shaped target.

For a target of dimension $d$, given a sequence of scale parameters, $\sigma_1,\dots,\sigma_d$ we consider the following four forms:
\begin{align*}
  \pi_G(x)
  &=
  \prod_{i=1}^d\frac{1}{\sigma_i \sqrt{2\pi}}\exp\left(-\frac{x_i^2}{2\sigma_i^2}\right),\\
  \pi_L(x)
  &=
  \prod_{i=1}^d\frac{1}{\sigma_i}\frac{\exp(x_i/\sigma_i)}{\{1+\exp(x_i/\sigma_i)\}^2},\\
  \pi_{SG}(x)
  &=
  \prod_{i=1}^d\frac{2}{\sigma_i \sqrt{2\pi}}\exp\left(-\frac{x_i^2}{2\sigma_i^2}\right)\Phi\left(\frac{\alpha x_i}{\sigma_i}\right),\\
  \pi_{MR}(x)
  &\propto
  \prod_{i=1}^{d/2}\exp\left\{-\frac{1}{2s_i^2}(x_{2i-i}-\sqrt{2}\beta s_i)^2\right\}
  \exp\left\{-\frac{1}{2}\left(x_{2i}-\frac{1}{\sqrt{2} s_i}\frac{x_{2i-1}^2}{1+x_{2i-i}^2/(4s_i^2)}\right)^2\right\},
\end{align*}
 where $\Phi$ is the distribution function of a $\mathsf{N}(0,1)$ variable, $\alpha=3$, $\beta=1$ and $s_i^2=99(i-1)/(d/2-1)+1$. 
The targets $\pi_G$, $\pi_L$ an $\pi_{SG}$ are products of one-dimensional Gaussian, logistic and skew-Gaussian distributions respectively. The potential surface of the target $\pi_{MR}$ is a modified version of the Rosenbrock function \cite[]{Rosenbrock1960}, a well-known, difficult target for optimisation algorithms, which is also a challenging target used to benchmark MCMC algorithms \cite[e.g.][]{PaWi2021,HengJacob2019}. The tails of the standard Rosenbrock potential increase quartically, but all algorithms which use a standard leapfrog step and Gaussian momentum are unstable in the tails of a target where the potential increases faster than quadratically. We have, therefore, modified the standard Rosenbrock target to keep the difficult banana shape whilst ensuring quadratic tails. 

For $\pi_G$, $\pi_L$ and $\pi_{SG}$, we denote the largest ratio of length scales by $\xi:=\max_{1\le i\le j\le d}\sigma_i/\sigma_j$ and define four different patterns between the lowest and highest scaling, depending on which of $\sigma_i$, $\sigma_i^2$, $1/\sigma_i^2$ or $1/\sigma_i$ increases approximately linearly with component number. To minimise the odd behaviour that HMC (but not AAPS or the no U-turn sampler) can exhibit when scalings are rational multiples of each other (a phenomenon which is rarely seen for targets in practice) we jitter the scales for all the intermediate components. Specifically, let $\uv$ be a vector with $v_1=0$, $v_d=1$ and for $i=2,\dots,d-1$, $v_i=(i-1+U_i)/(d-1)$ where $U_2,\dots,U_{d-1}$ are independent $\mathsf{Unif}(-0.5,0.5)$ variables. Then we define the following four progressions for $i=1,\dots,d$:
 SD: $\sigma_i=(\xi-1)v_i+1$; VAR: $\sigma_i^2=(\xi^2-1)v_i+1$; H: $1/\sigma_i^2=(1-1/\xi^2)v_i+1/\xi^2$; invSD: $1/\sigma_i=(1-1/\xi)v_i+1/\xi$.


A final target, $\pi_G^{RN}$, arises from an online comparison between HMC and the no U-turn sampler at\\
\url{https://radfordneal.wordpress.com/2012/01/27/evaluation-of-nuts-more-comments-on-the-paper-by-hoffman-and-gelman/}.

Figure \ref{fig.GaussSkGLSens} uses $\xi=20$ and repeats Figure \ref{fig.RosenSens} in $d=40$ for each of the main product target types: Gaussian, skew-Gaussian and logistic. It demonstrate the robustness of AAPS when compared with HMC and blurred HMC. Appendix \ref{sec.moreRobustnessFigures} contains more details on these experiments. Figure \ref{fig.GaussSkGLSensNUTS} in Appendix \ref{sec.moreRobustnessFigures} plots efficiency against step size when the popular No U-turn Sampler is used on the targets in Figures \ref{fig.RosenSens} and \ref{fig.GaussSkGLSens}. The peaks for the logistic and modified Rosenbrock targets are broad, suggesting some robustness; however, the peaks for the Gaussian and skew-Gaussian targets are much sharper.

\begin{figure}
\begin{center}
    \includegraphics[scale=0.33]{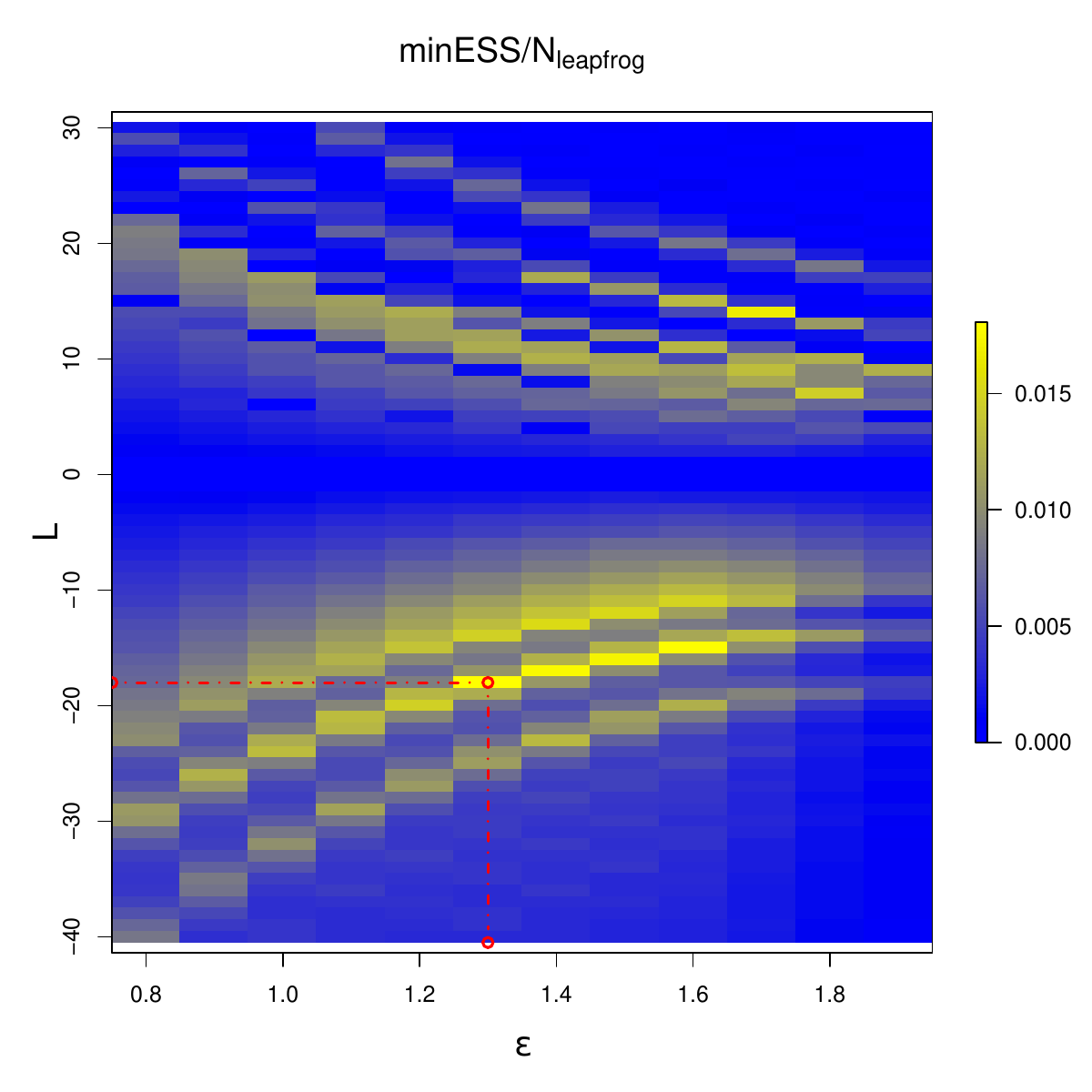}
    \includegraphics[scale=0.33]{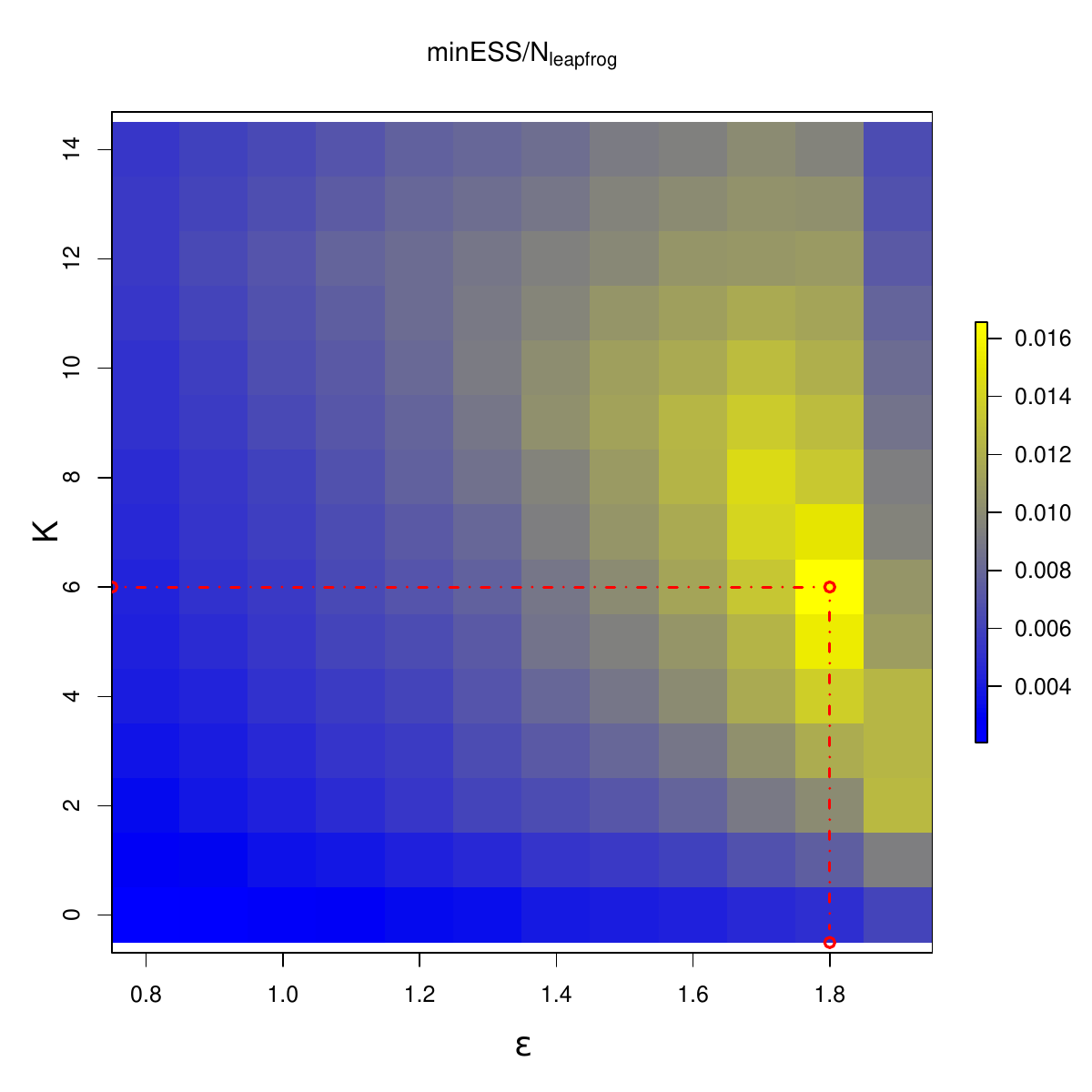}\\
    \includegraphics[scale=0.33]{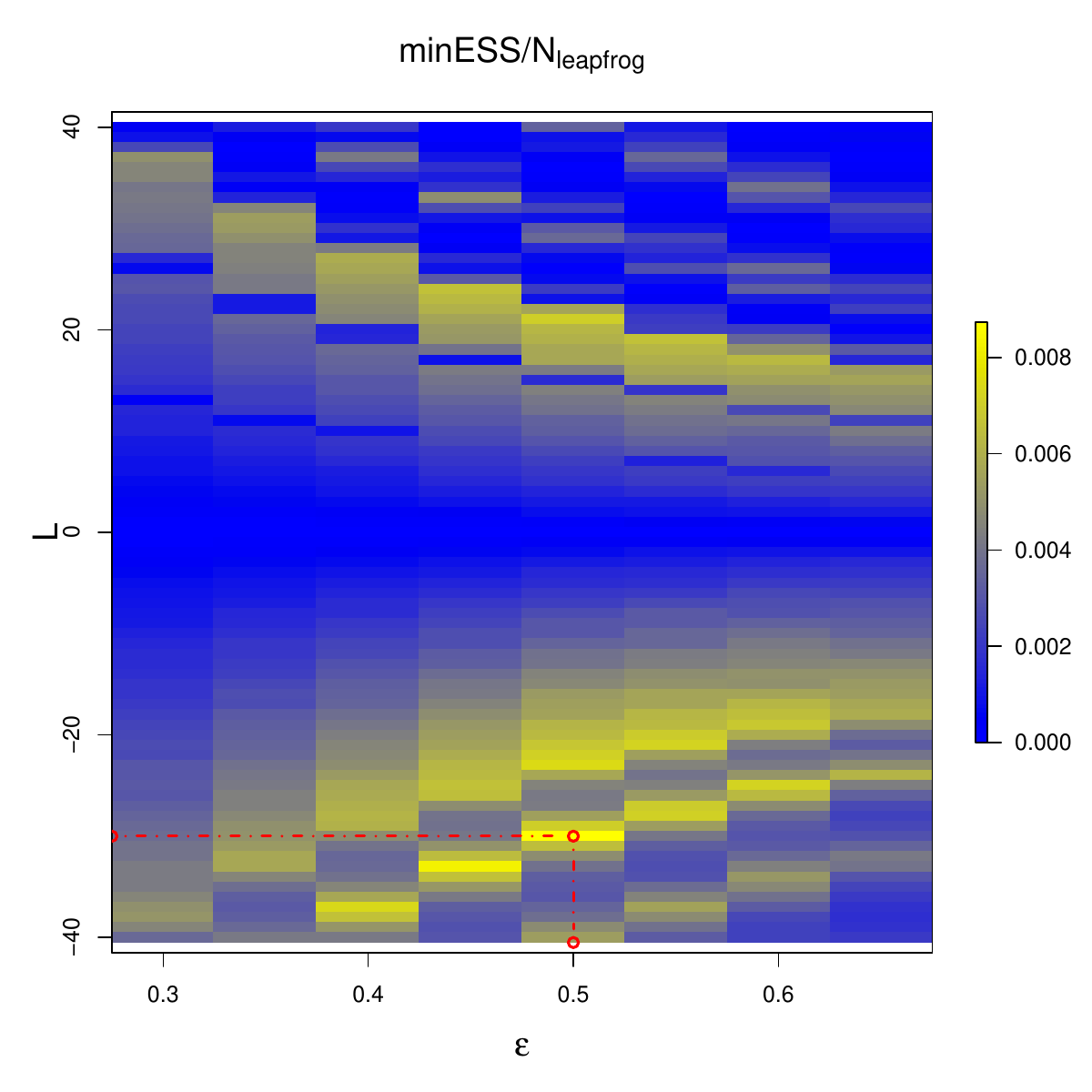}
    \includegraphics[scale=0.33]{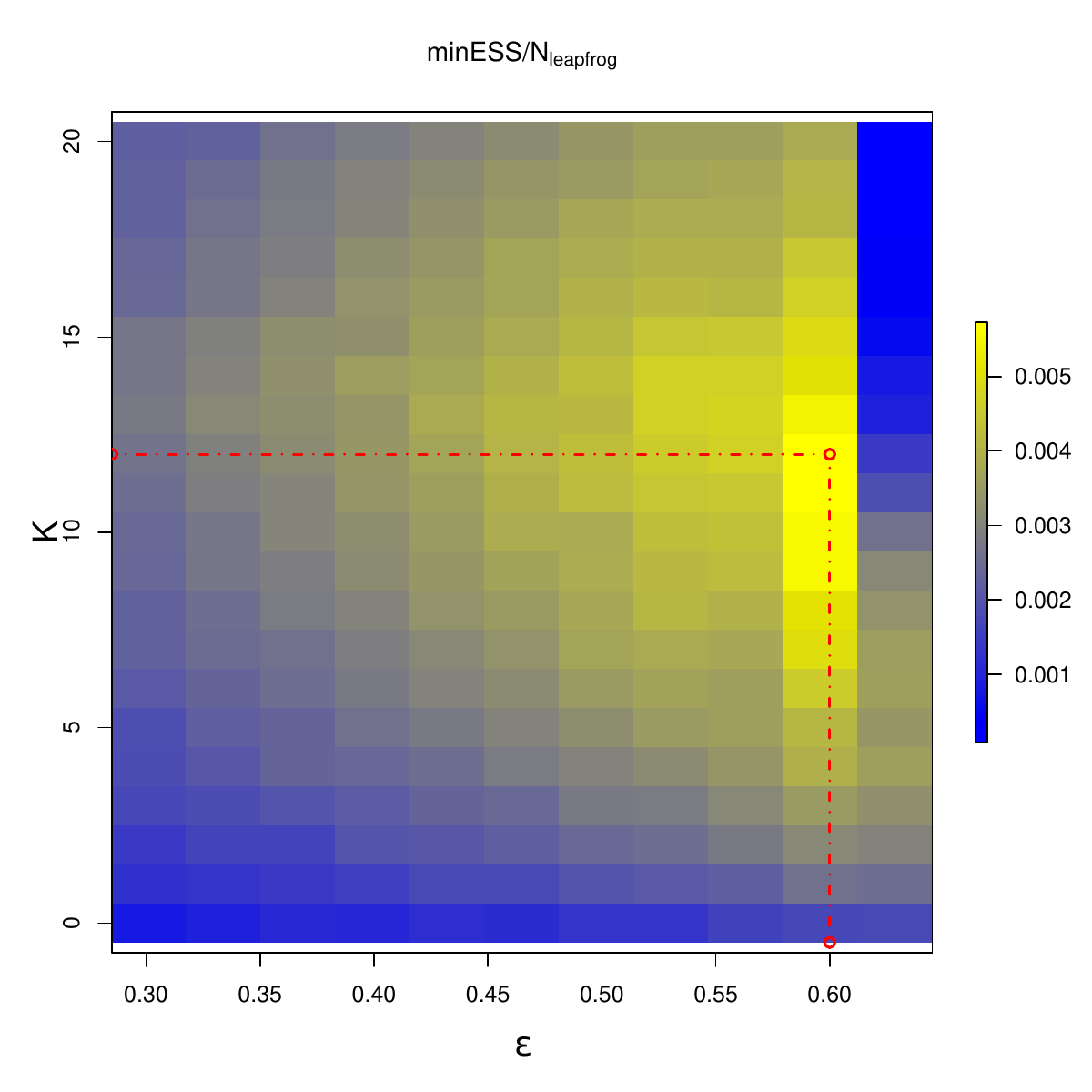}\\
    \includegraphics[scale=0.33]{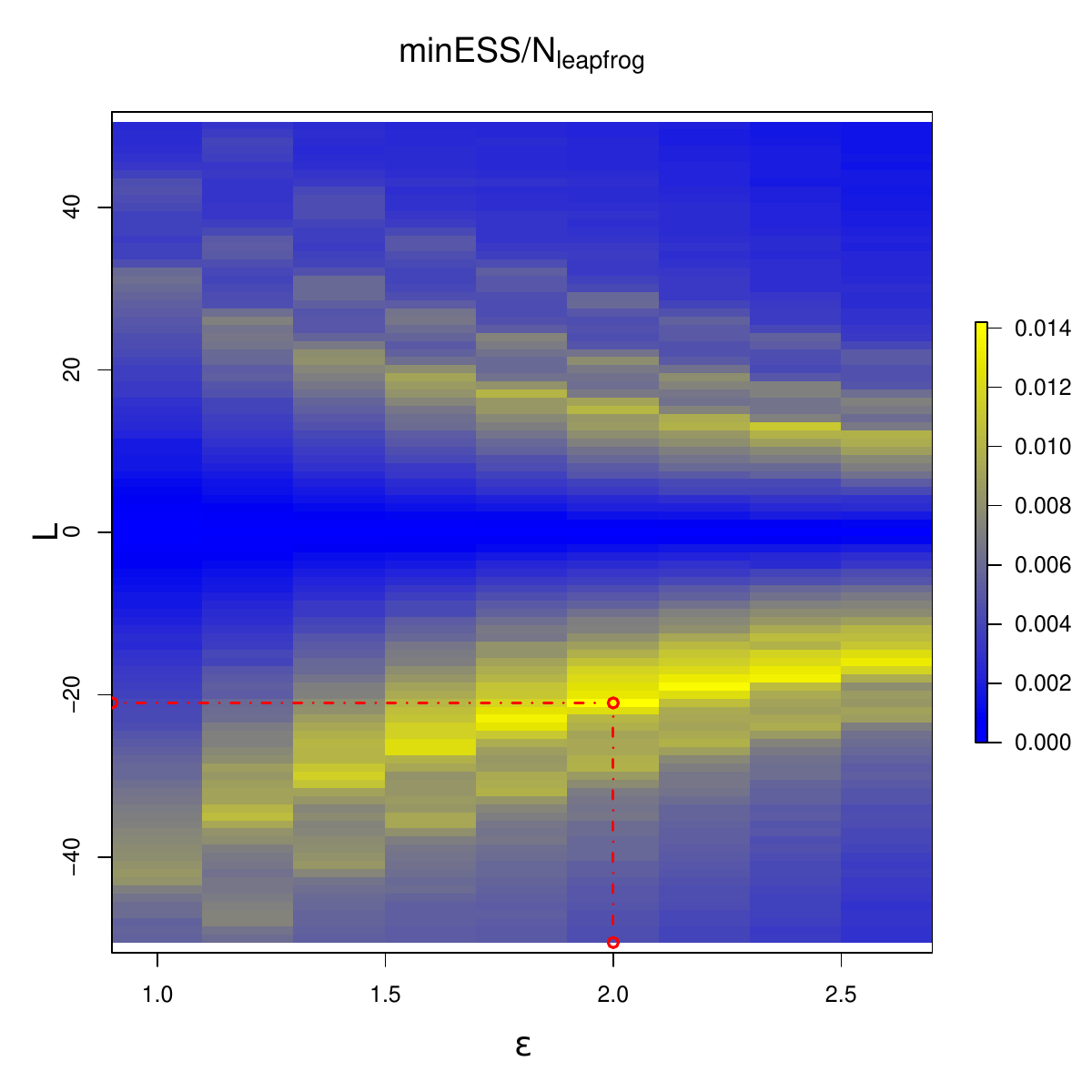}
    \includegraphics[scale=0.33]{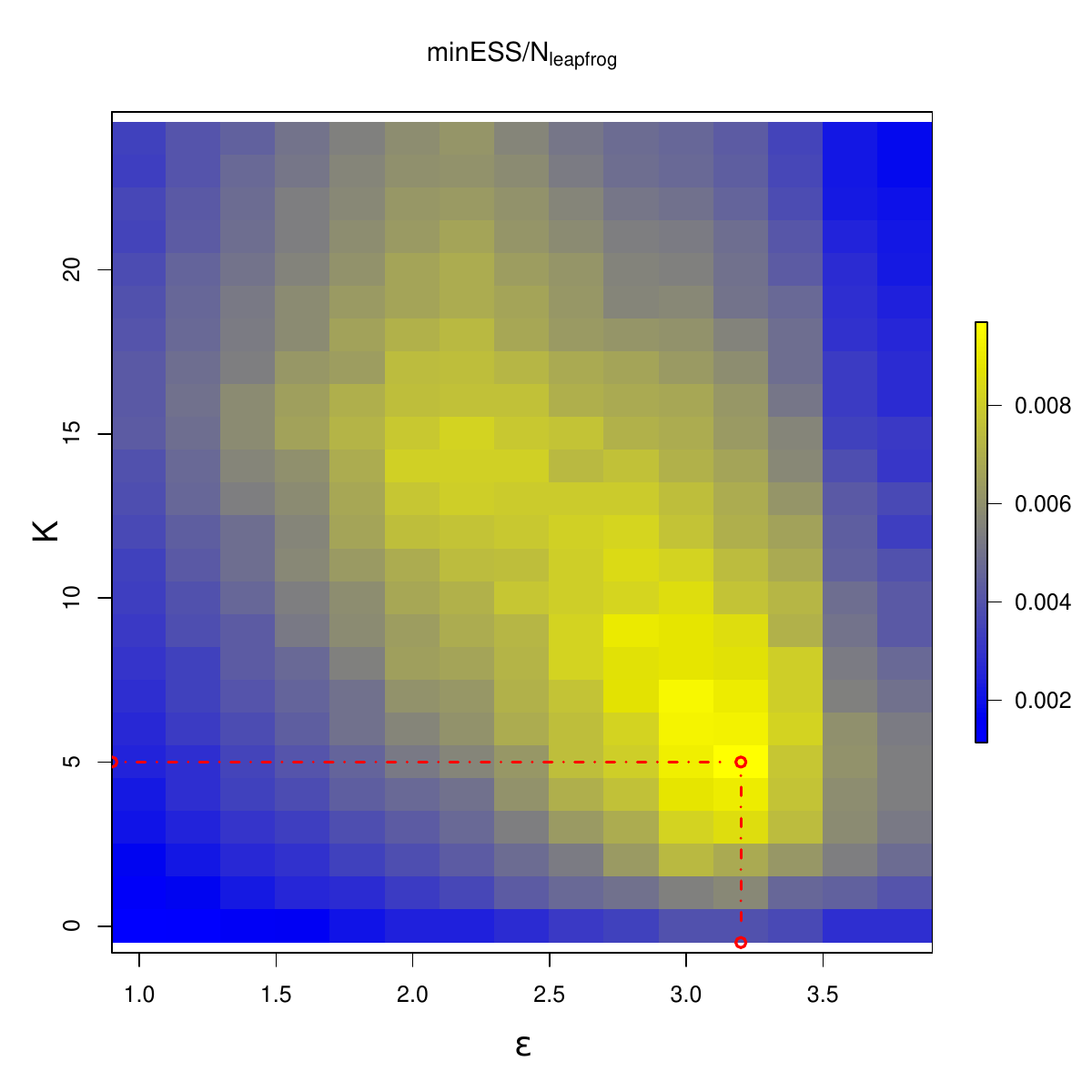}
    \end{center}
\caption{Efficiency, measured via \eqref{eqn.gen.eff}, as a function of the tuning parameters for $\pi_G^{VAR}$ (top), $\pi_{SG}^{VAR}$ (middle) and $\pi_L^{VAR}$ (bottom) all with $(\xi,d)=(20,40)$ as defined in Section \ref{sec.ToyTargets}. Left panels: HMC with positive $L$ values corresponding to the standard HMC algorithm and negative $L$ correspond to $|L|$ leapfrog steps of blurred HMC. Right panels: AAPS. Optimal tuning in red.
\label{fig.GaussSkGLSens}}
\end{figure}

We next investigate products of Gaussians using each of the four scale-parameter progressions with $\xi=20$. Dimension $d=40$ is high enough that for this and higher dimensions the components can be well approximated as arising from some continuous distributions and Corollary \ref{cor.EN} applies.

Table \ref{tab.ToyTargetRelEff_REV} shows the efficiencies of HMC, blurred HMC and the no U-turn sampler relative to that of $\AAPS$. Amongst the Gaussian targets, the relative performance of $\AAPS$ compared with HMC and NUTS is best for $\pi_G^{invSD}$ and $\pi_G^{H}$, which have just a few components with large scales and many components with small scales; weighting Scheme 3 ensures that the large components are explored preferentially. In contrast, the large number of components with large scales in $\pi_G^{VAR}$ leads to the worst relative performance of $\AAPS$; we, therefore, investigate this regime further using alternative component distributions, and choice of dimension and $\xi$. Across the range of targets, no algorithm is more than $1.7$ times as efficient as $\AAPS$. Empirical acceptance rates at the optimal parameter settings are provided in Appendix \ref{sec.TuningEpsilon}. All ESS estimates were at least $1000$.

\begin{table}
  \begin{center}
\caption{\label{tab.ToyTargetRelEff_REV}Relative efficiency compared with AAPS of HMC, blurred HMC (HMC-bl) and the no U-turn sampler (NUTS); raw efficiencies were taken as the minimum over the $d$ components of the number of effective samples per leapfrog step. The optimum was found by grid search. $^{(1)}$This is not a typographic error: the values differ at the 5th decimal place. $^{(2)}\xi$ for odd-numbered components of $\pi_{MR}$. $^{(3)}$The tuning surface for HMC was so uneven that were were unable to ascertain the optimal tuning parameters with any degree of certainty; instead we picked the least unreasonable of the combinations we tried.}
\begin{tabular}{l|rr|cccc}
  Target type&$d$&$\xi$&$\AAPS$&HMC&HMC-bl&NUTS\\
  \hline
  $\pi_G^{SD}$&40&20&1.000&0.722&0.718&1.182\\
  $\pi_G^{VAR}$&40&20&1.000&1.016&1.091&1.461\\
  $\pi_G^H$&40&20&1.000&$^{(1)}$0.162&0.644&0.392\\
  $\pi_G^{invSD}$&40&20&1.000&$^{(1)}$0.162&0.461&0.460\\
  \hline
      $\pi_{SG}^{VAR}$&40&20&1.000&1.253&1.528&1.618\\
  $\pi_{L}^{VAR}$&40&20&1.000&1.135&1.488&1.677\\
  $\pi_G^{VAR}$&100&20&1.000&0.657&1.020&1.378\\
  $\pi_G^{VAR}$&40&40&1.000&1.190&1.346&1.645\\
  \hline
   $\pi_{MR}$&20&$^{(2)}$10&1.000&1.647&1.582&0.728\\
  $\pi_{MR}$&40&$^{(2)}$10&1.000&1.045&1.166&0.873\\
  $\pi_{MR}$&100&$^{(2)}$10&1.000&0.770&0.970&1.079\\
  $\pi_{MR}$&400&$^{(2)}$10&1.000& 0.684 & 0.859 &0.963\\
  \hline
  $\pi_{G}^{RN}$&30&110&1.000&$^{(3)}$0.019&1.206&0.306
\end{tabular}
\end{center}
\end{table}

Table \ref{tab.ToyTargetRelEffRec_REV} of Appendix \ref{sec.useTuningAdvice} provides the equivalent efficiencies when the algorithms are tuned according to recommended guidelines, and shows AAPS remaining competitive with blurred HMC and the no U-turn sampler.

\subsection{Stochastic volatility model}
Consider the following model for zero-centred, Gaussian data
$y=(y_1,\dots,y_{T})$ where the variance depends on a zero-mean,
Gaussian AR(1) process started from stationarity \cite[e.g.][]{GirCal2011,WuSt2019}:
\begin{align*}
Y_t &\sim \mathsf{N}(0, \kappa^2\exp x_t),~t=1,\dots,T,\\
X_1 &\sim \mathsf{N}\left(0,\frac{\sigma^2}{1-\phi^2}\right),
~~~X_t|(X_{t-1}=x_{t-1})\sim \mathsf{N}(\phi x_{t-1},\sigma^2),~t=2,\dots,T.
\end{align*}
Parameter priors are $\pi_0(\kappa)\propto 1/\kappa$, $1/\sigma^2\sim \mathsf{Gamma}(10,0.05)$ and $(1+\phi)/2 \sim \mathsf{Beta}(20, 1.5)$, and we reparameterise to give a parameter vector in $\mathbb{R}^3$:
\[
\alpha=\log\frac{1+\phi}{1-\phi},
~~~
\beta=\log \kappa,
~~~\mbox{and}~~~
\gamma=2\log \sigma.
\]
As in \cite{GirCal2011} and \cite{WuSt2019} we generate $T=1000$ observations using parameters $\phi=0.98$, $\kappa=0.65$ and $\sigma=0.15$. We then apply blurred HMC, $\AAPS$ and the no U-turn sampler to perform inference on the $1003$-dimensional posterior. We ran $\AAPS$ using two different $K$ values, one found by optimising the choice of $(\epsilon,K)$ over a numerical grid and one by using the tuning mechanism mentioned in Section \ref{sec.tuning}. Tuning standard HMC (not blurred HMC) on such a high-dimensional target was extremely difficult; due to the algorithm's sensitivity to the integration time, we could not identify a suitable range for $L$. Widely used statistical packages such as Stan \cite[]{Stan2020} and PyMC3 \cite[]{SaWi2016} perform the blurring by default, and so we only present the results for HMC-bl. 

Each algorithm was then run for ten replicates of $10^5$ iterations using the optimal tuning parameters. Efficiency was calculated for each parameter, and the minimum efficiency over the latent variables and over all $1003$ components were also calculated. For each algorithm the mean and standard deviation (over the replicates) of these efficiencies were ascertained; Table \ref{Table.SV} reports these values normalised by the mean efficiency for $\AAPS$ for that parameter or parameter combination. Overall, on this complex, high-dimensional posterior, $\AAPS$ is slightly less efficient than blurred HMC, and slightly more efficient than the no U-turn sampler.

\begin{table}
\caption{Relative efficiency compared with $\AAPS^g$ ($\AAPS$ tuned using a grid) of $\AAPS^a$ ($\AAPS$ tuned using the advice from Section \ref{sec.tuning}), blurred HMC (HMC-bl) and the no U-turn sampler (NUTS) for the stochastic volatility model using $10$ replicates of $10^5$ iterations. Raw efficiencies were the number of effective samples per leapfrog step; these were then normalised by the mean efficiency from $\AAPS^g$; normalised standard deviations are reported in brackets.
\label{Table.SV}}
  \begin{center}
\begin{tabular}{c|c|c|c|c}
  Parameter & $\AAPS^g$ & $\AAPS^a$& HMC-bl & NUTS\\
  \hline
  $\alpha$&1.00 (0.12)&1.04 (0.19)&1.05 (0.18)&0.73 (0.25)\\
  $\beta$ &1.00 (0.09)&0.87 (0.11)&1.04 (0.13)&1.24 (0.29)\\
  $\gamma$&1.00 (0.03)&1.20 (0.05)&1.14 (0.07)&0.74 (0.19)\\
  $\min_{t\in\{1,\dots,T\}}\mathsf{ESS}(X_t)$     &1.00 (0.13)&0.89 (0.16)&1.07 (0.19)&0.78 (0.44)\\
  $\min_{\alpha,\beta,\gamma,t\in\{1,\dots,T\}}\mathsf{ESS}$&1.00 (0.05)&0.91 (0.12)&1.06 (0.11)&0.76 (0.21)\\
  \hline
  acc. rate ($\%$)&75.1&75.5&66.5&86.9
\end{tabular}
\end{center}
\end{table}

\subsection{Multimodality}
The $\AAPS$ algorithm with $K=0$ is close to reducible on a multimodal one-dimensional target, and this might lead to concerns about the algorithm's performance on multimodal targets in general. However, in $d$ dimensions, because $p^\top \nabla U(x)$ is a sum of $d$ components even with $K=0$, $\AAPS$ is not reducible on multimodal targets with $d>1$. This is illustrated in Appendix \ref{sec.bimodal}, which also details a short simulation study on three $40$-dimensional bimodal targets where $\AAPS$ is more efficient than the no U-turn sampler and is never less than two-thirds as efficient as blurred HMC. 

\section{Discussion}
\label{sec.Discussion}
We have presented the \emph{Apogee to Apogee Path Sampler} ($\AAPS$), and demonstrated empirically that it has a similar efficiency to HMC but is much easier to tune. From a current point, $\AAPS$ uses the leapfrog step to create a path consisting of a fixed number of \emph{segments}, it then proposes a point from these segments and uses an accept-reject step to ensure that it targets the intended distribution.

We investigated six possible mechanisms for proposing a point from the path, and for the numerical experiments we chose the probability of proposing a point to be proportional to the product of the extended target density at that point and the proposal's squared distance from the current point, which was possible with an $\mathcal{O}(1)$ memory cost. However, the flexibility in the proposal mechanism allows other possibilities such as a Mahalanobis distance based on an estimated covariance matrix, or $(||\xprop-\mu||^2-||\xcurr-\mu||^2)^2$ for some central point, $\mu$, with a similar motivation to the ChEEs diagnostic of \cite{HoRa2021}. Indeed, if any scalar or vector function, $f$, is of particular interest, then a proposal weighting of the form $||f(\xprop)-f(\xcurr)||^2$ could be used with a memory cost of $\mathcal{O}(1)$.

Choosing the current segment's position uniformly at random from the $K+1$ segments is not the only way to preserve detailed balance with respect to the intended target. For example, the current segment could be fixed as segment $0$ and proposals could only be made from segment $K$, a choice which bears some resemblance to the window scheme in \cite{Neal1992AnIA}; however, we found that this had a negative impact on the robustness of the efficiency to the choice of $K$ (see Appendix \ref{sec.moreRobustnessFigures}).

Because of its simplicity, many extensions to the $\AAPS$ algorithm are clear. For example, if the positions along the path are stored, then a delayed rejection step may increase the acceptance probabilities. A cheap surrogate for $\pi$ could be substituted within $\pitil$ in any weighting scheme. Indeed, given $c$, the randomly chosen offset of segment $0$, the next value could be chosen conditional on $\zcurr$ using any Markov kernel reversible with respect to $\pitil$ \cite[see also][]{neal2011ensemble}. A non-reversible version of the algorithm could set $K=1$ and always choose a path consisting of the current segment and the next segment forward; instead of completely refreshing momentum at each iteration, the momentum at the start of a new step could be a Crank-Nicolson perturbation of the momentum at the end of the previous step as in \cite{Horowitz1991}.
The properties of the leapfrog step required for the validity of AAPS are a subset of those required for HMC (see Section \ref{sec.HMC}), so any alternative momentum formulation \cite[e.g.][]{LivFauRob2019} or numerical integration scheme that can be used within HMC could also be used within AAPS.

$\AAPS$ with weighting Scheme 1 relates to HMC using windows of states \cite{Neal2011} but with $K$ defining the total number of (forward and backward) leapfrog steps taken rather than the number of additional segments. As mentioned in Section \ref{sec.intro} and shown in Corollary \ref{cor.EN}, the number of apogees is a more natural tuning parameter than an integration time as it relates to intrinsic properties of the target: rescaling all co-ordinates by a constant factor would not change the optimal $K$. 


\section*{Acknowledgements}
Work by all authors was supported by EPSRC grant EP/P033075/1. We are grateful to two anonymous reviewers for comments and suggestions that have materially improved the article.\\

Code available from \url{https://github.com/ChrisGSherlock/AAPS}

\bibliographystyle{rss}
\bibliography{../aaps_refs}

\begin{thebibliography}{24}
\expandafter\ifx\csname natexlab\endcsname\relax\def\natexlab#1{#1}\fi
\expandafter\ifx\csname url\endcsname\relax
  \def\url#1{\texttt{#1}}\fi
\expandafter\ifx\csname urlprefix\endcsname\relax\def\urlprefix{URL: }\fi

\bibitem[{Beskos et~al.(2013)Beskos, Pillai, Roberts, Sanz-Serna and
  Stuart}]{BePi2013}
Beskos, A., Pillai, N., Roberts, G., Sanz-Serna, J.-M. and Stuart, A. (2013)
  Optimal tuning of the hybrid {M}onte {C}arlo algorithm.
\newblock \textit{Bernoulli}, \textbf{19}, 1501--1534.
\newblock \urlprefix\url{http://www.jstor.org/stable/42919328}.

\bibitem[{Bou-Rabee and Sanz-Serna(2017)}]{rHMC}
Bou-Rabee, N. and Sanz-Serna, J.~M. (2017) {R}andomized {H}amiltonian {M}onte
  {C}arlo.
\newblock \textit{The Annals of Applied Probability}, \textbf{27}, 2159--2194.
\newblock \urlprefix\url{http://www.jstor.org/stable/26361544}.

\bibitem[{Brooks et~al.(2011)Brooks, Gelman, Jones and Meng}]{MCMChandbook}
Brooks, S., Gelman, A., Jones, G.~L. and Meng, X.-L. (eds.) (2011)
  \textit{Handbook of {M}arkov chain {M}onte {C}arlo}.
\newblock Chapman \& Hall/CRC Handbooks of Modern Statistical Methods. Boca
  Raton, FL: CRC Press.

\bibitem[{Duane et~al.(1987)Duane, Kennedy, Pendleton and Roweth}]{DuKe1987}
Duane, S., Kennedy, A., Pendleton, B.~J. and Roweth, D. (1987) Hybrid {M}onte
  {C}arlo.
\newblock \textit{Physics Letters B}, \textbf{195}, 216--222.
\newblock
  \urlprefix\url{https://www.sciencedirect.com/science/article/pii/037026938791197X}.

\bibitem[{Gilks et~al.(1996)Gilks, Richardson and Spiegelhalter}]{GiRiSp1996}
Gilks, W.~R., Richardson, S. and Spiegelhalter, D.~J. (1996) \textit{{M}arkov
  Chain {M}onte {C}arlo in practice}.
\newblock London, UK: Chapman and Hall.

\bibitem[{Girolami and Calderhead(2011)}]{GirCal2011}
Girolami, M. and Calderhead, B. (2011) Riemann manifold {L}angevin and
  {H}amiltonian {M}onte {C}arlo methods.
\newblock \textit{Journal of the Royal Statistical Society: Series B
  (Statistical Methodology)}, \textbf{73}, 123--214.
\newblock
  \urlprefix\url{https://rss.onlinelibrary.wiley.com/doi/abs/10.1111/j.1467-9868.2010.00765.x}.

\bibitem[{Heng and Jacob(2019)}]{HengJacob2019}
Heng, J. and Jacob, P.~E. (2019) {Unbiased Hamiltonian {M}onte {C}arlo with
  couplings}.
\newblock \textit{Biometrika}, \textbf{106}, 287--302.
\newblock \urlprefix\url{https://doi.org/10.1093/biomet/asy074}.

\bibitem[{Hoffman et~al.(2021)Hoffman, Radul and Sountsov}]{HoRa2021}
Hoffman, M., Radul, A. and Sountsov, P. (2021) An adaptive-{MCMC} scheme for
  setting trajectory lengths in {H}amiltonian {M}onte {C}arlo.
\newblock In \textit{Proceedings of The 24th International Conference on
  Artificial Intelligence and Statistics} (eds. A.~Banerjee and K.~Fukumizu),
  vol. 130 of \textit{Proceedings of Machine Learning Research}, 3907--3915.
  PMLR.
\newblock \urlprefix\url{https://proceedings.mlr.press/v130/hoffman21a.html}.

\bibitem[{Hoffman and Gelman(2014)}]{HoGe2014}
Hoffman, M.~D. and Gelman, A. (2014) The {N}o-{U}-{T}urn sampler: adaptively
  setting path lengths in hamiltonian {M}onte {C}arlo.
\newblock \textit{Journal of Machine Learning Research}, \textbf{15},
  1593--1623.

\bibitem[{Horowitz(1991)}]{Horowitz1991}
Horowitz, A.~M. (1991) A generalized guided {M}onte {C}arlo algorithm.
\newblock \textit{Physics Letters B}, \textbf{268}, 247--252.
\newblock
  \urlprefix\url{https://www.sciencedirect.com/science/article/pii/0370269391908125}.

\bibitem[{Livingstone et~al.(2019)Livingstone, Faulkner and
  Roberts}]{LivFauRob2019}
Livingstone, S., Faulkner, M.~F. and Roberts, G.~O. (2019) {Kinetic energy
  choice in {H}amiltonian/hybrid {M}onte {C}arlo}.
\newblock \textit{Biometrika}, \textbf{106}, 303--319.
\newblock \urlprefix\url{https://doi.org/10.1093/biomet/asz013}.

\bibitem[{{Mackenze}(1989)}]{Mack1989}
{Mackenze}, P.~B. (1989) {An improved hybrid {M}onte {C}arlo method}.
\newblock \textit{Physics Letters B}, \textbf{226}, 369--371.

\bibitem[{Neal(1992)}]{Neal1992AnIA}
Neal, R.~M. (1992) An improved acceptance procedure for the hybrid {M}onte
  {C}arlo algorithm.
\newblock \textit{Journal of Computational Physics}, \textbf{111}, 194--203.

\bibitem[{Neal(2011{\natexlab{a}})}]{neal2011ensemble}
--- (2011{\natexlab{a}}) {MCMC} using ensembles of states for problems with
  fast and slow variables such as {G}aussian process regression.

\bibitem[{Neal(2011{\natexlab{b}})}]{Neal2011}
--- (2011{\natexlab{b}}) {MCMC} using {H}amiltonian dynamics.
\newblock In \textit{Handbook of {M}arkov chain {M}onte {C}arlo} (eds.
  S.~Brooks, A.~Gelman, G.~Jones and X.-L. Meng), chap.~5, 113--162. CRC press.

\bibitem[{Pagani et~al.(2021)Pagani, Wiegand and Nadarajah}]{PaWi2021}
Pagani, F., Wiegand, M. and Nadarajah, S. (2021) An n-dimensional {R}osenbrock
  distribution for {Markov chain Monte Carlo} testing.
\newblock \textit{Scandinavian Journal of Statistics}.
\newblock
  \urlprefix\url{https://onlinelibrary.wiley.com/doi/abs/10.1111/sjos.12532}.

\bibitem[{Pompe et~al.(2020)Pompe, Holmes and Latuszynski}]{PoHo2020}
Pompe, E., Holmes, C. and Latuszynski, K. (2020) {A framework for adaptive MCMC
  targeting multimodal distributions}.
\newblock \textit{The Annals of Statistics}, \textbf{48}, 2930 -- 2952.
\newblock \urlprefix\url{https://doi.org/10.1214/19-AOS1916}.

\bibitem[{Roberts and Rosenthal(2001)}]{RobRos2001}
Roberts, G.~O. and Rosenthal, J.~S. (2001) Optimal scaling for various
  {M}etropolis-{H}astings algorithms.
\newblock \textit{Statistical Science}, \textbf{16}, 351--367.

\bibitem[{Rosenbrock(1960)}]{Rosenbrock1960}
Rosenbrock, H.~H. (1960) {An Automatic Method for Finding the Greatest or Least
  Value of a Function}.
\newblock \textit{The Computer Journal}, \textbf{3}, 175--184.
\newblock \urlprefix\url{https://doi.org/10.1093/comjnl/3.3.175}.

\bibitem[{Salvatier et~al.(2016)Salvatier, Wiecki and Fonnesbeck}]{SaWi2016}
Salvatier, J., Wiecki, T.~V. and Fonnesbeck, C. (2016) {Probabilistic
  programming in Python using PyMC3}.
\newblock \textit{PeerJ Computer Science}, \textbf{2}, e55.

\bibitem[{Sherlock and Roberts(2009)}]{SheRob2009}
Sherlock, C. and Roberts, G. (2009) {Optimal scaling of the random walk
  Metropolis on elliptically symmetric unimodal targets}.
\newblock \textit{Bernoulli}, \textbf{15}, 774 -- 798.
\newblock \urlprefix\url{https://doi.org/10.3150/08-BEJ176}.

\bibitem[{{Stan Development Team}(2020)}]{Stan2020}
{Stan Development Team} (2020) Stan modeling language users guide and reference
  manual.
\newblock \urlprefix\url{http://mc-stan.org/}.
\newblock Version 2.28.

\bibitem[{Wu et~al.(2019)Wu, Stoehr and Robert}]{WuSt2019}
Wu, C., Stoehr, J. and Robert, C.~P. (2019) {H}amiltonian {M}onte {C}arlo by
  learning leapfrog scale.

\bibitem[{Ylvisaker(1965)}]{Ylvisaker1965}
Ylvisaker, N.~D. (1965) The expected number of zeros of a stationary {G}aussian
  process.
\newblock \textit{Annals of Mathematical Statistics}, \textbf{36}, 1043--1046.

\end{thebibliography}

\appendix


\section{The leapfrog step}
\label{sec.leapfrog}
From a current $z_0=(x_0,p_0)$ the leapfrog step takes this to $z_1=(x_1,p_1)$ as follows:
\begin{align*}
  p'\gets p-\frac{1}{2}\epsilon \nabla U(x_0),~~~
  x_1\gets x_0+\epsilon M^{-1}p',~~~
  p_1\gets p'-\frac{1}{2}\epsilon \nabla U(x_1).
\end{align*}
The transformation is clearly deterministic, and starting at $(x_1,-p_1)$ and applying the three steps leads to $(x_0,-p_0)$. 
Further, each of the three transformations has a Jacobian of $1$, so the composition of the three also has a Jacobian of $1$.

\section{Weighting Schemes 1-3 with $\mathcal{O}(1)$ memory cost}
\label{sec.memOone}

As in Section \ref{sec.chooseWeight} we label the
points in $\cS_{a:b}(\zcurr)$ from furthest back in time to furthest forward in time as $z_B,\dots,z_0,\dots,z_F$. If $z'$ has been chosen from $\{z_0,\dots, z_F\}$ with a probability of $w(\zcurr,z')/\sum_{i=0}^F w(\zcurr,z_i)$ and $z''$ has been chosen from $\{z_{B},\dots z_{-1}\}$ with a probability of $w(\zcurr,z'')/\sum_{i=B}^{-1} w(\zcurr,z_i)$, then we obtain the desired proposal by setting $\zprop=z'$ with a probability of \[
\frac{\sum_{i=0}^F w(\zcurr,z_i)}{\sum_{i=B}^F w(\zcurr,z_i)},
\]
and otherwise setting $\zprop=z''$. We, therefore, use the following algorithm that samples $z'$ from an ordered list of vectors $z_c,\dots,z_e$ with a probability proportional to $w(z,z')$ and without storing the intermediate values; we then apply this function to $\{z_0,\dots,z_F\}$ and $\{z_{-1},\dots,z_{B}\}$.

\begin{enumerate}
\item Set $z'=z_c$.
\item For $j$ in $\{c+1,\dots,e\}$
  \begin{itemize}
\item With probability $\frac{w(z,z_j)}{\sum_{i=c}^j w(z,z_i)}$ set $z'=z_j$, else leave $z'$ unchanged.
    \end{itemize}
  \end{enumerate}

As required,
\begin{align*}
\Prob{Z'=z_k}
&=
\frac{w(z,z_k)}{\sum_{i=c}^k w(z,z_i)}
\prod_{j=k+1}^e \left(1-\frac{w(z,z_j)}{\sum_{i=c}^j w(z,z_i)}\right)\\
&=
\frac{w(z,z_k)}{\sum_{i=c}^k w(z,z_i)}
\prod_{j=k+1}^e \frac{\sum_{i=c}^{j-1} w(z,z_i)}{\sum_{i=c}^j w(z,z_i)}
=
\frac{w(z,z_k)}{\sum_{i=c}^{e} w(z,z_i)}.
\end{align*}

This allows us to propose $\zprop$ with the correct probability for the first five weighting schemes, as well as to evaluate the numerator term $\sum_{z\in \cS_{a:b}(\zcurr)}w(\zcurr,z)$ in \eqref{eq.AAPSalpha}. However, if all of the individual $z$ and evaluations of $\pitil(z)$ were forgotten, then they would need to be recalculated from scratch to evaluate the denominator term $\sum_{z\in \cS_{a:b}(\zcurr)}w(\zprop,z)$, entailing a repeat of all $F+B$ expensive leapfrog steps. We now show how this can be evaluated using only $3\times d$ summary statistics when $x\in \mathbb{R}^d$.

Firstly,
\begin{align}
  \nonumber
\sum_{i=c}^e \pitil(z_i)||x_i-\xprop||^2
&=
\sum_{i=c}^e \pitil(z_i) \sum_{k=1}^d(x_{i,k}-\xprop_{k})^2\\
\nonumber
&=
\sum_{k=1}^d \sum_{i=c}^e \pitil(z_i) x_{i,k}^2-\pitil(z_i)2x_{i,k}\xprop_k+\pitil(z_i)(\xprop_{k})^2\\
&=
\sum_{k=1}^d T_k^{(2)}-2\xprop_kT_k^{(1)}+(\xprop_k)^2 T_k^{(0)},
\label{eq.interimTot}
\end{align}
where for $k=1,\dots,d$,
\[
T_k^{(0)}=\sum_{i=c}^e \pitil(z_i),
~~~
T_k^{(1)}=\sum_{i=c}^e \pitil(z_i) x_{i,k},
~~~\mbox{and}~~~
T_k^{(2)}=\sum_{i=c}^e \pitil(z_i) x^2_{i,k}.
\]
We create these $3d$ summary statistics for $\{z_0,\dots,z_F\}$ and for $\{z_{-1},\dots,z_B\}$, apply \eqref{eq.interimTot} in each case, and add the two final totals together.

In the case of Scheme 2, $\pitil(z_i)$ is replaced with $1$, and for Scheme 1, the acceptance probability is $1$, so \eqref{eq.interimTot} is unnecessary.

\section{Proofs of results in Section \ref{sec.GPlimit}}

\subsection{Proof of Theorem \ref{thrm.GP}}
\label{sec.prove.GP.thrm}
\begin{proof}
When $U$ has the form \eqref{eq.product.target} and $\rho\equiv \mathsf{N}(0,I_d)$, the components are independent and, from Hamilton's equations \eqref{eq.HamiltonEqns}, they  evolve independently. We, therefore, start by considering an individual component, initiated at a scalar position $X_0\sim \pi_\nu$ and momentum  $P_0\sim \rho_1$, and with $(x_t,p_t)=\phi(t;x_0,p_0)$ as in Section \ref{sec.Hdynamics}.

The deterministic map $(x_0,p_0)\to (x_t,p_t)$ has a Jacobian of $1$ so that the density of $(X_t,P_t)$ is $f(x_t,p_t)=\pitil_\nu(\phi(x_t,p_t;-t))=\pitil_\nu(x_0,p_0)$. However, the conservation of energy under Hamiltonian dynamics is equivalent to
  \[
  \frac{\md}{\md t} \pitil_\nu(x_t,p_t)=0,
  \]
  so $f\equiv \pitil_\nu$ and $(X_t,P_t)$ is stationary.

  Now consider the contribution to the dot product from this single component:
\[
D(t;x_0,p_0,\nu):=\left.\frac{\md U}{\md x}\right|_{x_t} p_t,
=
-\frac{\md p}{\md t} p_t
=
-\frac{1}{2}\frac{\md}{\md t} p_t^2,
\]
by Hamilton's equations \eqref{eq.HamiltonEqns}. Hence
\begin{equation}
  \label{eq.GPproveExpect}
  \Expects{(X_0,P_0)\sim \pitil_\nu}{D(t;X_0,P_0,\nu)}
  =
  -\frac{1}{2}\Expects{(X_0,P_0)\sim \pitil_\nu}{\frac{\md}{\md t} P_t^2}
  =
  -\frac{1}{2}\frac{\md}{\md t} \Expects{(X_0,P_0)\sim \pitil_\nu}{P_t^2}
  =0,
\end{equation}
by stationarity, giving the expectation in \eqref{eqn.SGPfull}. 

Next, by Hamilton's equations \eqref{eq.HamiltonEqns}
\[
\sqrt{\nu} g'(\sqrt{\nu}x)=\frac{\md U}{\md x}=-\frac{\md p}{\md t}
=
-\frac{\md^2 x}{\md t^2},
\]
so
\[
x_t=\frac{1}{\sqrt{\nu}}a(\sqrt{\nu}t;\sqrt{\nu}x_0,p_0)
~~~\mbox{and}~~~
p_t=a'(\sqrt{\nu}t;\sqrt{\nu}x_0,p_0).
\]
Hence, we may write
\begin{align*}
D(t;x_0,p_0,\nu)
=
\left.\frac{\md U}{\md x}\right|_{x_t}p_t 
=
\sqrt{\nu} ~g'(\sqrt{\nu}x_t)~a'(\sqrt{\nu}t;\sqrt{\nu}x_0,p_0).
\end{align*}


By the stationarity of $(X_t,P_t)$ when $(X_0,P_0)\sim \pitil_\nu$, the joint distribution of $(X_s,P_s,X_t,P_t)$ for two times $s$ and $t$ is the same as that of $(X_0,P_0,X_{t-s},P_{t-s})$. Moreover, $(x_{t},p_{t})=\phi(t;x_0,p_0)=\phi(t-s;x_s,p_s)$. In particular, $p_{t}$ from a start at $(x_0,p_0)$ is identical to $p_{t-s}$ from a start at $(x_s,p_s)$; \emph{i.e.}, $a'(\sqrt{\nu}t;\sqrt{\nu}x_0,p_0)=a'(\sqrt{\nu}(t-s);\sqrt{\nu}x_s,p_s)$.

Define $V(s,t;\nu):=\Covs{(X_0,P_0)\sim\pitil_\nu}{D(s;X_0,P_0),D(t;X_0,P_0)}$. Since $\Expect{D}=0$,
\begin{align*}
  V(s,t;\nu)&=
  \nu \Expect{g'(\sqrt{\nu} X_s)a'(\sqrt{\nu}s;\sqrt{\nu}X_0,P_0)~g'(\sqrt{\nu} X_t)a'(\sqrt{\nu}t;\sqrt{\nu}X_0,P_0)}\\
  &=
  \nu \Expect{g'(\sqrt{\nu}X_s)a'(0;\sqrt{\nu}X_s,P_s)~g'(\sqrt{\nu} X_t)a'(\sqrt{\nu}(t-s);\sqrt{\nu}X_s,P_s)}\\
  &=
  \nu \Expect{g'(\sqrt{\nu}X_0)a'(0;\sqrt{\nu}X_0,P_0)~g'(\sqrt{\nu} X_{t-s})a'(\sqrt{\nu}(t-s);\sqrt{\nu}X_0,P_0)}\\
  &=
  \nu \Expect{g'(X_0^*)a'(0;X^*_0,P_0)~g'\left(\sqrt{\nu} X_{t-s}\right)a'\left(\sqrt{\nu}(t-s);X^*_0,P_0\right)}\\
  &=
  \nu \Expect{g'(X_0^*)P_0~g'\left(\sqrt{\nu} X_{t-s}\right)a'\left(\sqrt{\nu}(t-s);X^*_0,P_0\right)},  
\end{align*}
where $X_0^*=\sqrt{\nu}X_1\sim \pi_1$. 

Then, by the strong law of large numbers,
\begin{align*}
V^{(d)}(s,t)&:=
  \Cov{D^{(d)}(s;X_0^{(d)},P_0^{(d)}),D(t;X_0^{(d)},P_0^{(d)})}\\
&=\frac{1}{d}
\sum_{i=1}^d
\nu_i^{(d)}~ \Expects{(X_0,P_0)\sim \pitil_1}{g'(X_0)P_0~g'\left(\sqrt{\nu_i^{(d)}} X_{t-s}\right)a'\left(\sqrt{\nu_i^{(d)}}(t-s);X_0,P_0\right)}\\
&\stackrel{a.s.}{\rightarrow}
\Expects{\nu\sim \mu,(X_0,P_0)\sim \pitil_1}{\nu g'(X_0)P_0~g'\left(\sqrt{\nu} X_{t-s}\right)a'\left(\sqrt{\nu}(t-s);X_0,P_0\right)}\\
&=V(t-s),
\end{align*}
as defined in \eqref{eq.SGPcov}.

Take $\delta$ to be the minimum of the $\delta$ in \eqref{eq.finite.moment.one} and that in \eqref{eq.finite.moment.two}, and define
\[
b_\delta:=\Expects{X\sim \pi_1}{|g'(X)|^{2+\delta}}\Expects{P\sim \rho_1}{|P|^{2+\delta}},
\]
which is finite by \eqref{eq.finite.moment.one} and since $P$ is Gaussian.
By the stationarity of $(X_t,P_t)$,
\[
\Expect{|D(t;X_0,P_0)|^{2+\delta}}
=
\Expect{|D(0;X_0,P_0)|^{2+\delta}}
=\nu^{1+\delta/2}b_\delta.
\]
Further, for any finite integer $n$ and time points $t_1,\dots,t_n$ with $D_{1:n}=(D(t_1;X_0,P_0),\dots,D(t_n;X_0,P_0))^\top$
\begin{align*}
||D_{1:n}||^{2+\delta}&=n^{1+\delta/2}\left(\frac{1}{n}\sum_{j=1}^nD(t_j;X_0,P_0)^2\right)^{1+\delta/2}\\
&\le
n^{1+\delta/2}\left(\frac{1}{n}\sum_{j=1}^n|D(t_j;X_0,P_0)|^{2+\delta}\right)
=
n^{\delta/2}\sum_{j=1}^n|D(t_j;X_0,P_0)|^{2+\delta},
\end{align*}
by Jensen's inequality. Thus, 
\[
\Expect{||D_{1:n}||^{2+\delta}}
\le
n^{1+\delta/2}
\nu^{1+\delta/2}b_\delta.
\]
In particular, therefore, for any $\epsilon>0$,
\[
\Expect{||D_{1:n}||^2~1(||D_{1:n}||\ge \epsilon \sqrt{d})}
\le
\frac{1}{\epsilon^\delta d^{\delta/2}}
n^{1+\delta/2}\nu^{1+\delta/2}b_\delta.
\]
We have $d$ contributions as above, $D_{i,1:n}^{(d)}$, $i=1,\dots,d$, which we abbreviate to $D_{i}^{(d)}$. By the assumption that $V$ is positive definite, the left-hand side of the condition for the Lindeberg-Feller Central Limit Theorem reduces to
\[
\frac{1}{d}\sum_{i=1}^d \Expect{||D^{(d)}_i||^21(||D_i^{(d)}||\ge \epsilon\sqrt{n})}
\le
\frac{1}{d}\sum_{i=1}^d\frac{1}{\epsilon^{\delta} d^{\delta/2}}
n^{1+\delta/2}\nu_i^{1+\delta/2}b_\delta\rightarrow 0
\]
as required, subject to condition \eqref{eq.finite.moment.two}.

Hence, for any integer $n$, all $n$-dimensional distributions tend to a multivariate Gaussian with covariances given according to $V$, and the result follows.
\end{proof}

\begin{remark}
This proof would follow through for a non-Gaussian  momentum formulation subject to a moment condition on $p_0$.
\end{remark}

\subsection{Proof of Corollary \ref{cor.EN}}
\label{sec.prove.GP.corollary}
\begin{proof}
From \eqref{eq.ODE}  $g'(\sqrt{\nu} X_t)=-a''(\sqrt{\nu} t; X_0, P_0)$, so the covariance function is equivalent to
\[
V(t)=-\Expect{\nu g'(X_0)P_0 a''(\sqrt{\nu} t; X_0,P_0)a'(\sqrt{\nu} t;X_0,P_0)}.
\]
Thus
\begin{align*}
  V''(t)
  &=
  -\Expect{\nu^2 g'(X_0)P_0 \left\{a^{iv}(\sqrt{\nu} t; X_0,P_0)a'(\sqrt{\nu} t;X_0,P_0)+3 a''(\sqrt{\nu} t;X_0,P_0) a'''(\sqrt{\nu} t;X_0,P_0)\right\}},
\end{align*}
and
\begin{align*}
  V''(0)
  &=
  -\Expect{\nu^2 g'(X_0)P_0 \left\{a^{iv}(0; X_0,P_0)a'(0;X_0,P_0)+3 a''(0;X_0,P_0) a'''(0;X_0,P_0)\right\}}\\
  &=
  \Expect{\nu^2 g'(X_0)P_0 \left\{a^{iv}(0; X_0,P_0)P_0+3 g'(X_0) a'''(0;X_0,P_0)\right\}}.
\end{align*}

The variance of the process is
\[
V(0)=\Expects{\mu}{\nu}\Expects{\pi_1}{g'(X)^2}.
\]
The process, $D_*=\Dtil/\sqrt{V(0)}$ has a variance of $1$, and its covariance function has a second derivative at $0$ of
\[
V_*''(0)=\frac{  \Expect{\nu^2 g'(X_0)P_0 \left\{a^{iv}(0; X_0,P_0)P_0+3 g'(X_0) a'''(0;X_0,P_0)\right\}}}{\Expect{\nu}\Expect{g'(X)}^2}.
\]

The number of zeroes of $D_*$ is the same as the number of zeros of $\Dtil$.
 Thus, the expected number of zeros of $\Dtil$ over a time $T$ is
\[
\frac{T}{2\pi}\sqrt{-V_*''(0)}
\propto
T\sqrt{\frac{\Expect{\nu^2}}{\Expect{\nu}}}
=
T\sqrt{\Expect{\nu}}\times \sqrt{\frac{\Expect{\nu^2}}{\Expect{\nu}^2}},
\]
as required.
\end{proof}

\section{Robustness of efficiency to tuning choices}
\label{sec.moreRobustnessFigures}
Figure \ref{fig.GaussSkGLSens} is analogous to Figure \ref{fig.RosenSens} but for three other $40$-dimensional targets. These are products of one-dimensional distributions, respectively, Gaussian, logistic and skew-Gaussian, $\pi_G^{VAR}$, $\pi_L^{VAR}$ and $\pi_{SG}^{VAR}$ with $\xi=20$ as described in Section \ref{sec.ToyTargets}. For each target, over all components, the smallest scale parameter is $\sigma_1=1$. For a Gaussian target this implies that the leapfrog step is only stable provided $\epsilon <2\sigma_1=2$ \cite[e.g.][]{Neal2011}. We, therefore, stop at $\epsilon=1.9$, and even here, undesirable behaviour can be seen in Figure \ref{fig.GaussSkGLSens} (top row). For a skew-Gaussian distribution with $\sigma_1=1$ and $\alpha=3$, the density of the narrowest component in the left tail is
\[
2\phi(x) \Phi(3x)\approx \frac{2}{3x} \phi(x)\phi(3x) \propto \frac{1}{x}\phi(\sqrt{10}x).
\]
So, in the left tail, it is approximately equivalent to a Gaussian with $\sigma_1=1/\sqrt{10}$. Thus stability might only be hoped for when $\epsilon<2/\sqrt{10}$, as born out in Figure \ref{fig.GaussSkGLSens} (middle). The plots all show that for any sensible $\epsilon$, $\AAPS$ performance is much more robust to the choice of $K$ than HMC (whether blurred or not) is to the choice of $L$, or equivalently, $T$.

Figure \ref{fig.GaussSkGLSensNUTS} plots the efficiency of the No U-turn Sampler (NUTS) as a function of $\epsilon$ for the same four targets as in Figures \ref{fig.RosenSens} and \ref{fig.GaussSkGLSens}. The broad peak for the logistic and modified Rosenbrock targets suggests some robustness to the choice of $\epsilon$; however, the peak is much sharper for the Gaussian and skew-Gaussian targets.

\begin{figure}
\begin{center}
    \includegraphics[scale=0.27]{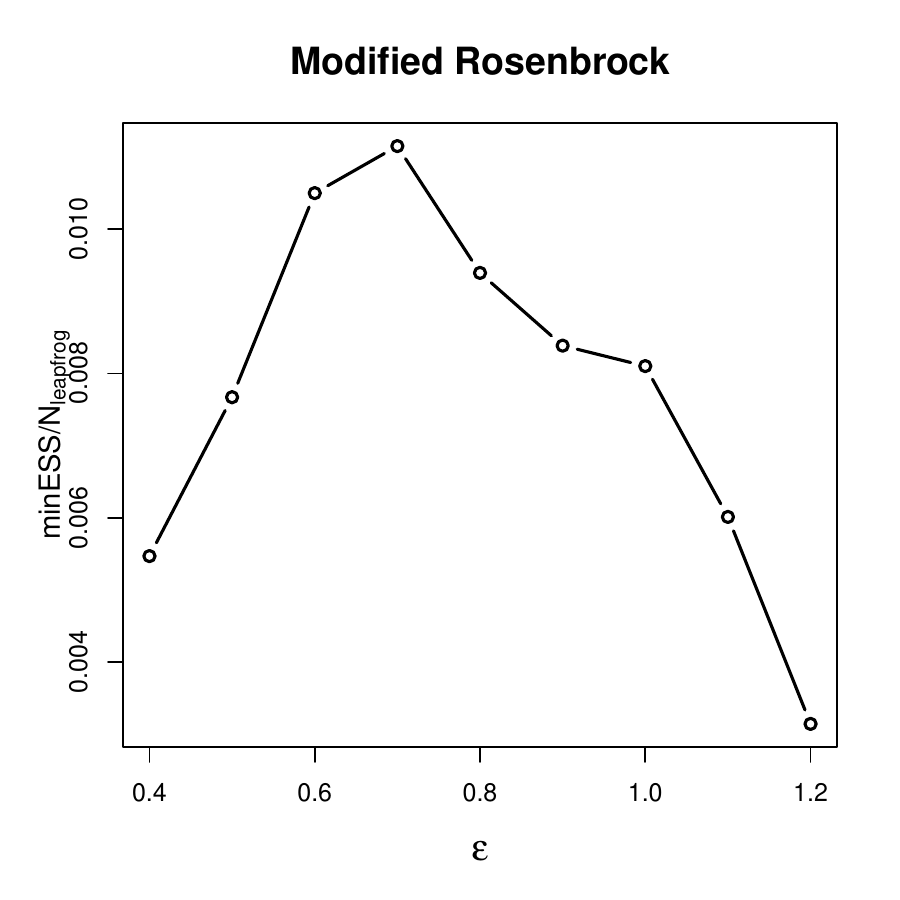}
    \includegraphics[scale=0.27]{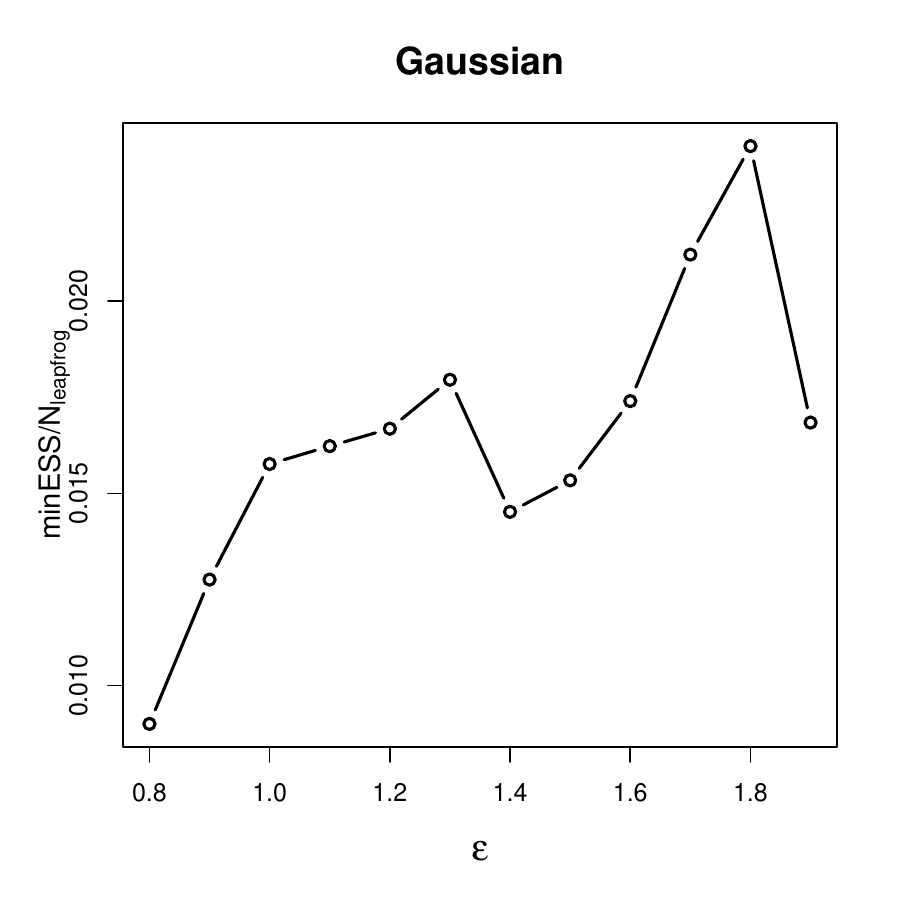}
    \includegraphics[scale=0.27]{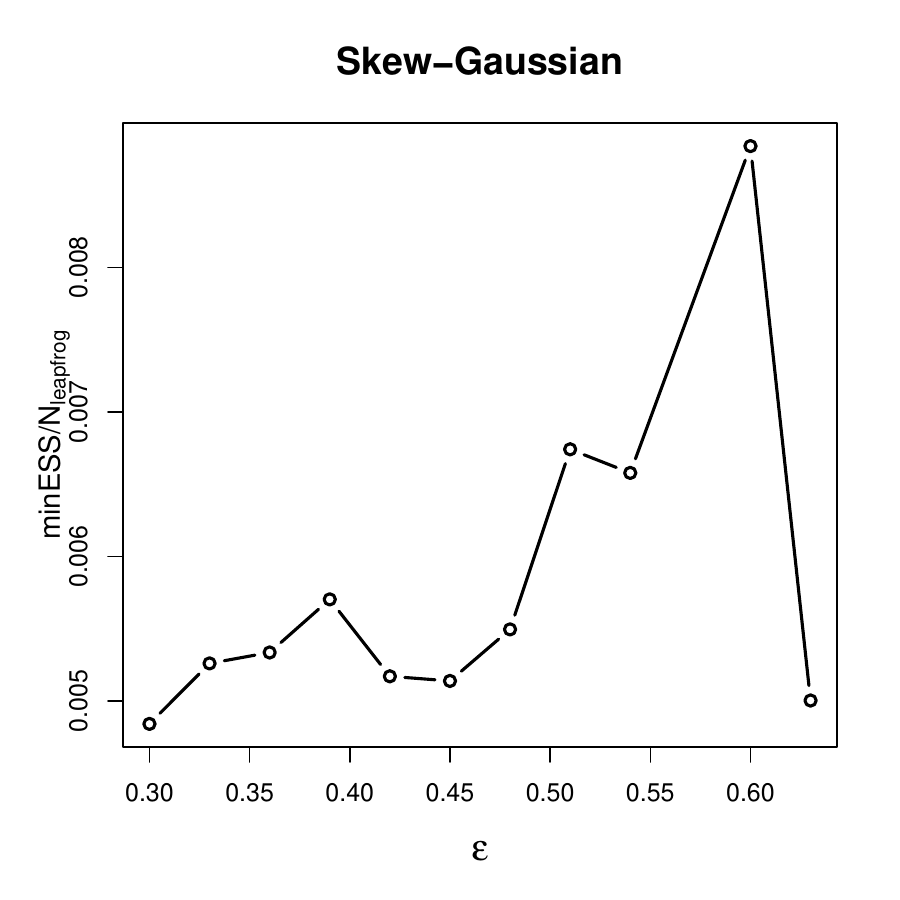}
    \includegraphics[scale=0.27]{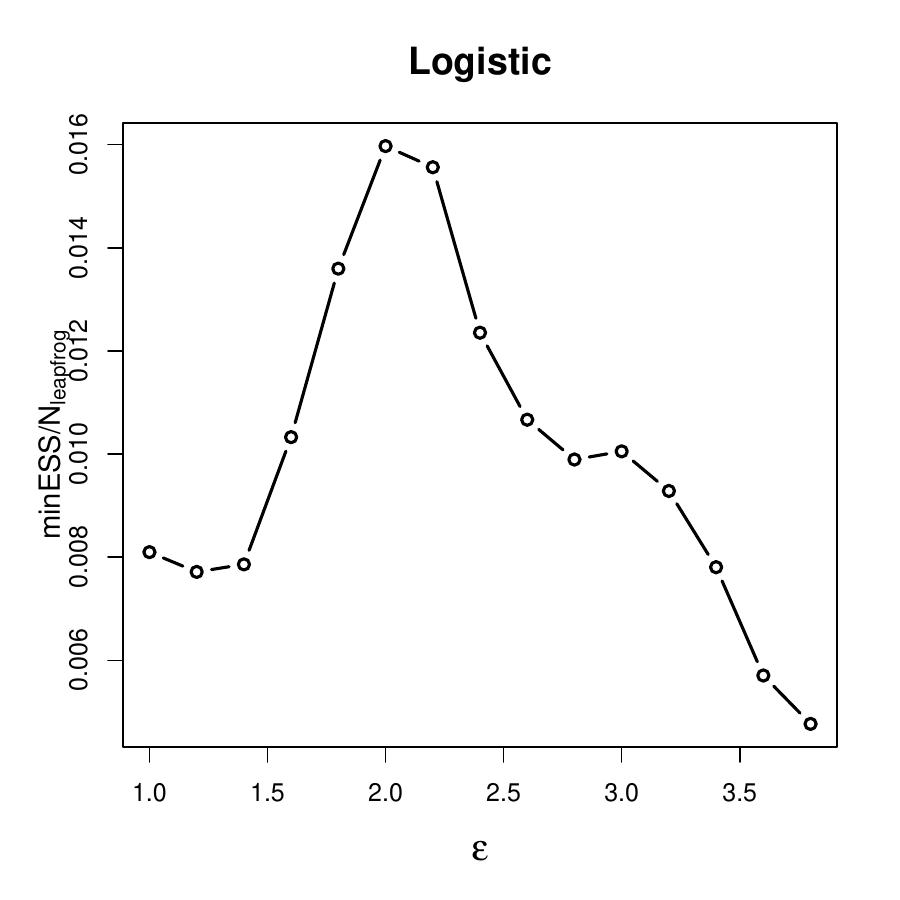}
    \end{center}
\caption{Efficiency, measured via \eqref{eqn.gen.eff}, as a function of the step size, $\epsilon$, for (left to right) $\pi_{MR}$, $\pi_G^{VAR}$, $\pi_{SG}^{VAR}$ and $\pi_L^{VAR}$, all with $d=40$ and with $\xi=(10,20,20,20)$, as defined in Section \ref{sec.ToyTargets}.
\label{fig.GaussSkGLSensNUTS}}
\end{figure}

Figure \ref{fig.sensDiffSchemes} compares four different versions of AAPS for the modified Rosenbrock target of Figure \ref{fig.RosenSens}. The left panels propose using Scheme 3 and the right panels use Scheme 1. The top panels are standard AAPS, whereas the bottom panels fix $c=0$ and always propose from segment $K$. For standard AAPS, Scheme 1 shares the robustness of Scheme 3, and in both cases Scheme 1 is less efficient at its optimum. However, whichever weighting scheme is used, the version of AAPS which only proposes from segment $K$ is much less robust than standard AAPS because of the limited portion of the path from which proposals can arise.

\begin{figure}
\begin{center}
  \includegraphics[scale=0.39]{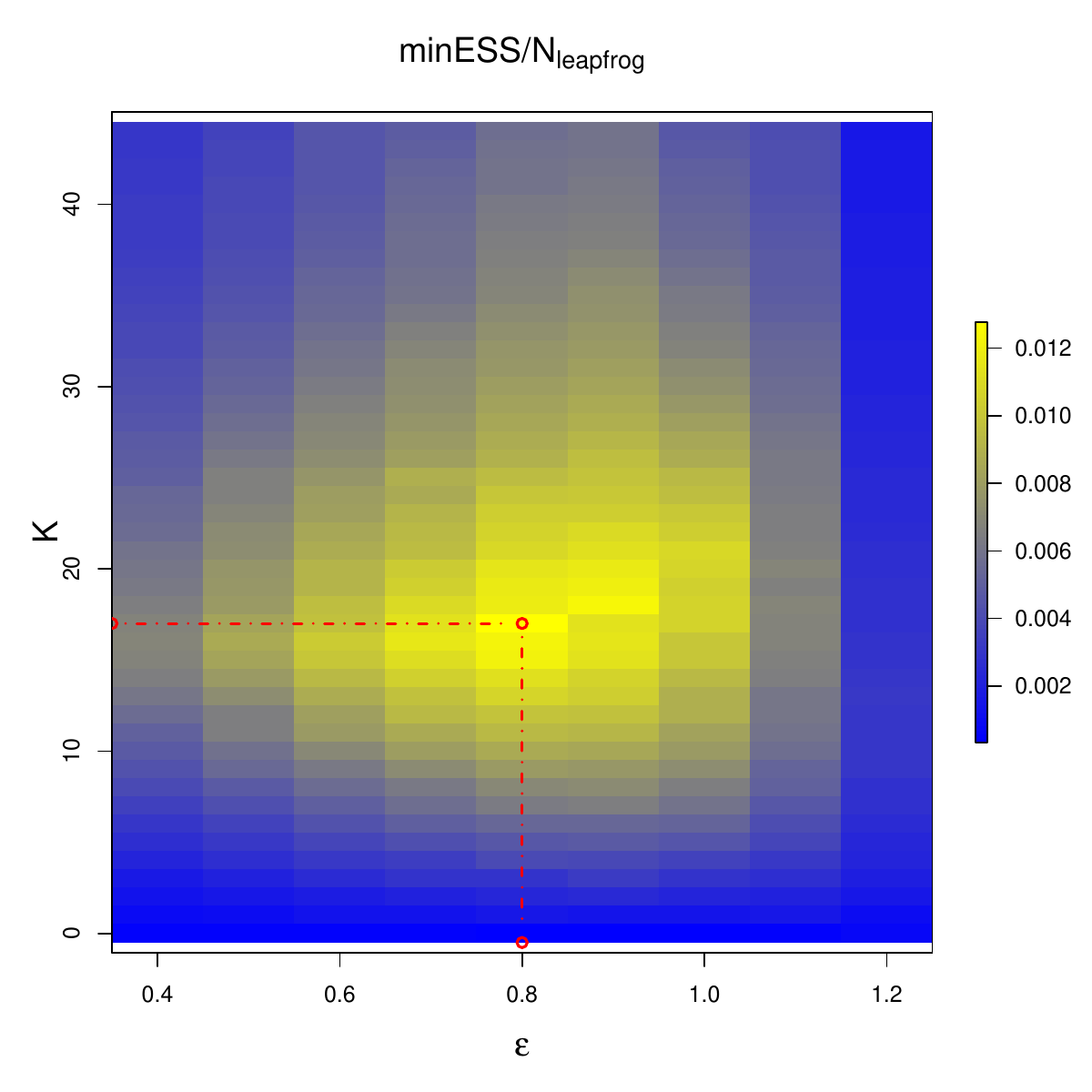}
      \includegraphics[scale=0.39]{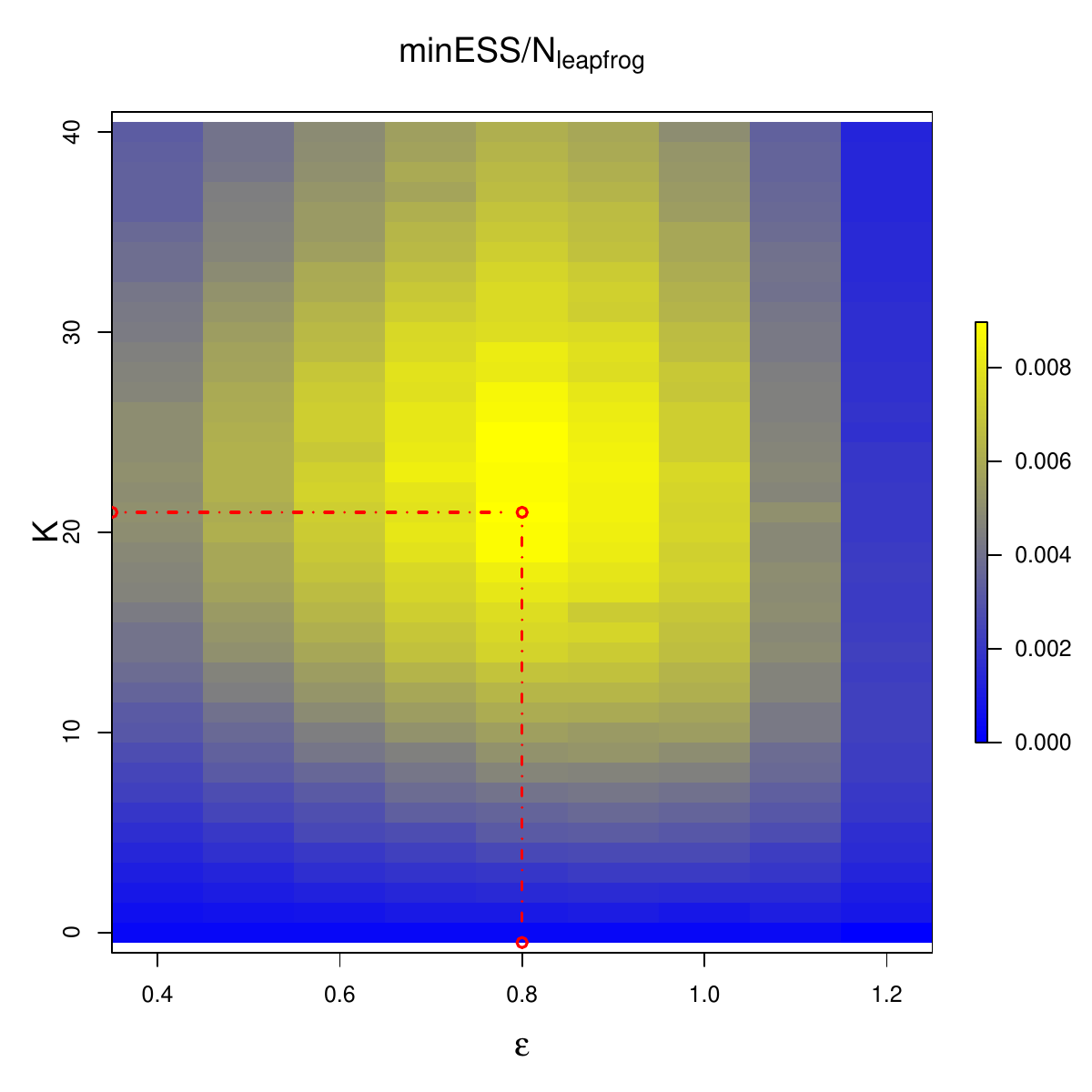}\\
  \includegraphics[scale=0.39]{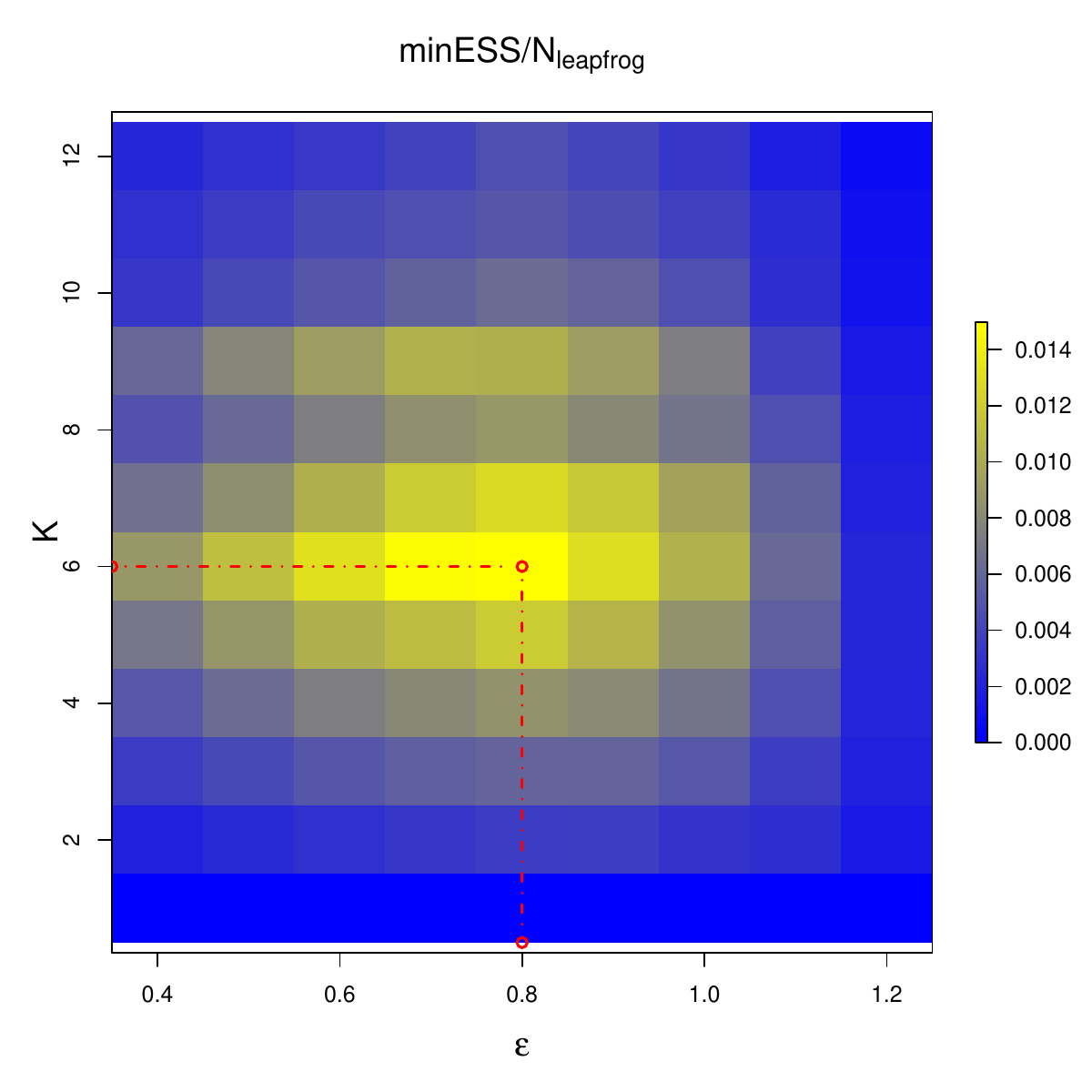}
      \includegraphics[scale=0.39]{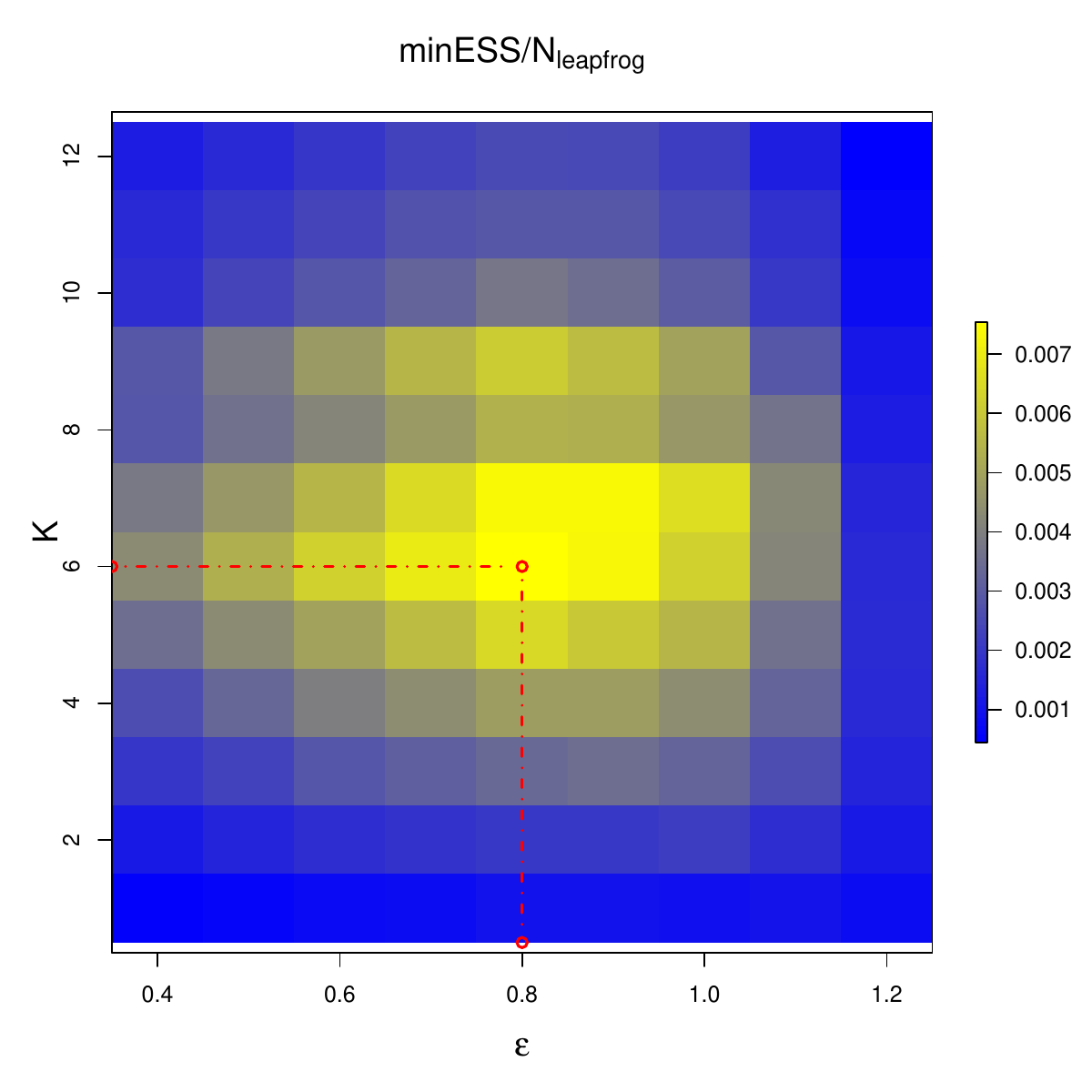}
    \end{center}
\caption{Efficiency at exploring the modified-Rosenbrock target of Figure \ref{fig.RosenSens} as a function of $\epsilon$ and $K$ for four different versions of the AAPS algorithm. Top-left: AAPS with Scheme 3 (repeating the right panel of Figure \ref{fig.RosenSens}); Top right: AAPS with Scheme 1. The bottom-left and bottom right panels also correspond to schemes 3 and 1 respectively but use a version of AAPS where $c=0$ and proposals are only made from the $K$th segment. 
\label{fig.sensDiffSchemes}}
\end{figure}

\section{Tuning $\epsilon$}
\label{sec.TuningEpsilon}
As $\epsilon \downarrow 0$, for weighting Schemes 2 and 3 (and Schemes 4 and 5), the acceptance rate approaches a constant, which is purely a function of the typical differences in position along the set of segments of the current point and of a proposal chosen according to squared jumping distance. This can be seen, for example, in the right panel of Figure \ref{fig.compareWeights}.

\begin{figure}
\begin{center}
    \includegraphics[scale=0.42]{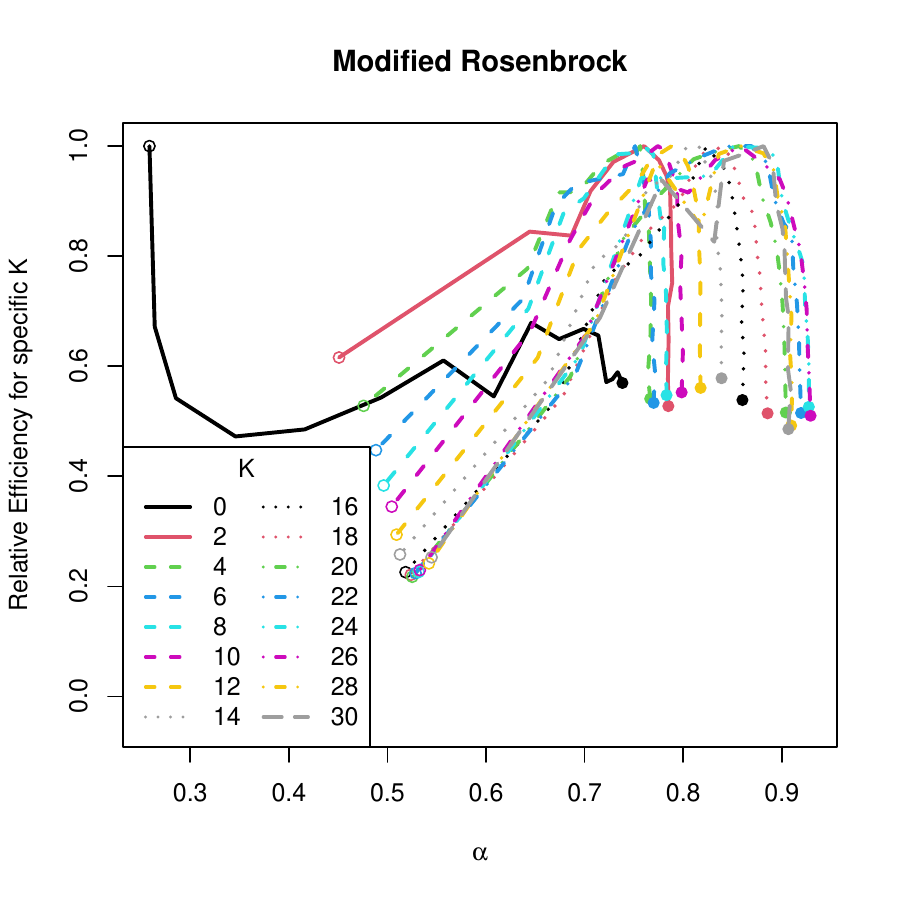}
    \includegraphics[scale=0.42]{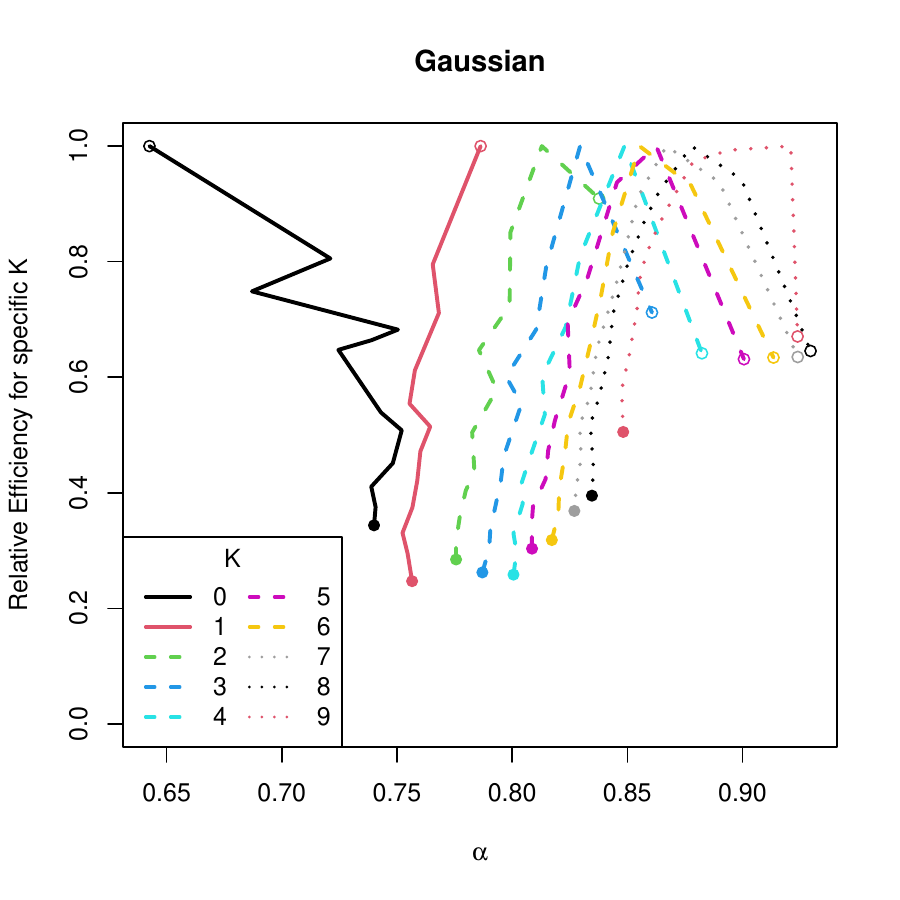}\\
    \includegraphics[scale=0.42]{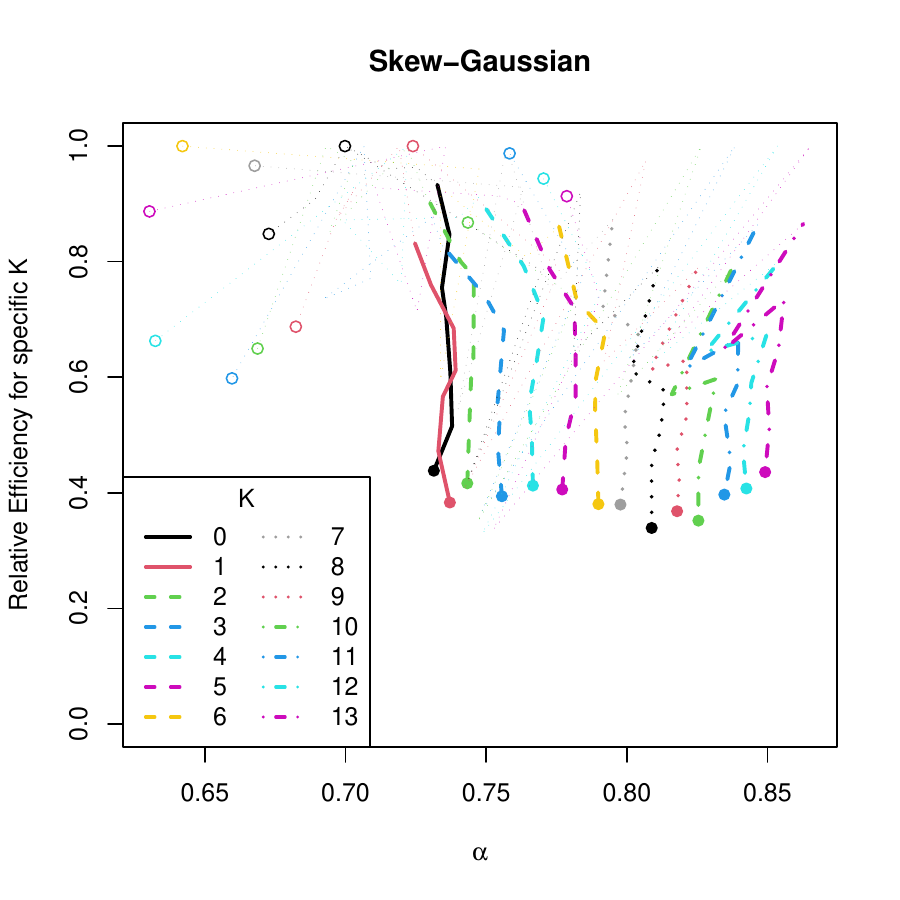}
    \includegraphics[scale=0.42]{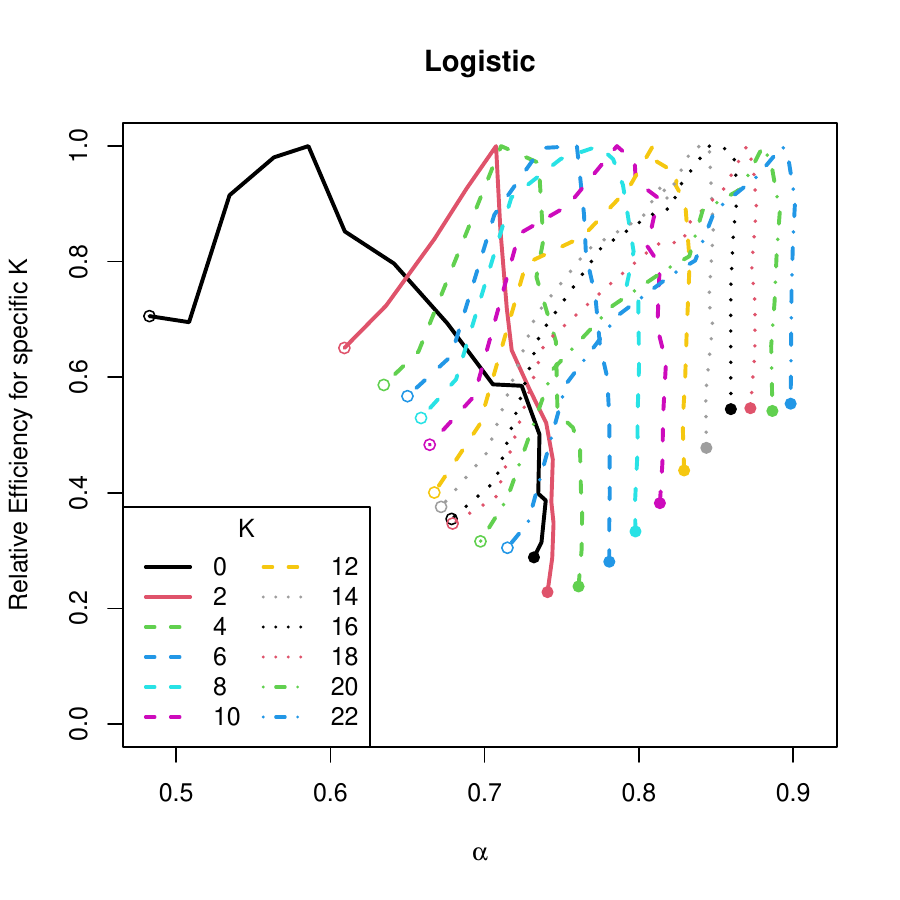}
    \end{center}
\caption{Relative efficiency (efficiency over the maximum obtained over all $\epsilon$ for that target and choice of $K$) plotted against empirical acceptance rate as $\epsilon$ varies. Targets are: the modified Rosenbrock distribution of Figure \ref{fig.RosenSens} (top left), and the Gaussian (top right), skew-Gaussian (bottom left) and logistic (bottom right) of Figure \ref{fig.GaussSkGLSens}. Each solid circle corresponds to the smallest vallue of $\epsilon$ used and each open circle to the largest. Lines on the skew-Gaussian plot are rendered thinner when $\epsilon>2/\sqrt{10}$ (see Appendix \ref{sec.moreRobustnessFigures}).
\label{fig.RelEffVsAccRate}}
\end{figure}

Figure \ref{fig.RelEffVsAccRate} corresponds to the AAPS runs in Figures \ref{fig.RosenSens} and \ref{fig.GaussSkGLSens} but shows efficiency as a function of empirical acceptance rate. As $\epsilon$ increases the vertical climb of each line shows efficiency increasing with barely a change in acceptance rate. When $K>0$ at least $60$\% of the optimal efficiency is achieved in this regime. For the logistic and modified-Rosenbrock examples, further increases in $\epsilon$ bring a decrease in acceptance rate; the skew-Gaussian acceptance rate shows some instability before it starts decreasing -- this instability is a consequence of the light tails, described in Appendix \ref{sec.moreRobustnessFigures}, and in practice one would choose a different kinetic energy formula \cite[]{LivFauRob2019}. For the Gaussian example the acceptance rate actually increases slightly before it decreases. This is because increasing $\epsilon$ causes some true apogees to be missed. This happens to an extent in all four examples; however with the Gaussian example it is so marked that for large $\epsilon$ the mean number of leapfrog steps per iteration actually increases as $\epsilon$ increases (\emph{e.g.}, when $\epsilon$ changes from $1.8$ to $1.9$ the mean number of leapfrog steps increases by around $10\%$), stabilising the weighted sums of $\pitil$ still further.

Table \ref{tab.numAcc} corresponds to Table \ref{tab.ToyTargetRelEff_REV}; it gives the empirical acceptance rates at the optimal parameter settings (found via a grid search) for each of the four algorithms and for the set of targets described in Section \ref{sec.ToyTargets}. Table \ref{tab.numAcc} also provides empirical estimates of the limiting acceptance rates as $\epsilon\downarrow 0$. The absolute discrepancies between the limiting and optimal acceptance rates are all below $5\%$. We recommend choosing the largest $\epsilon$ that leads to an acceptance rate within $3\%$ of the limiting value; this errs on the side of caution, as a lower discrepancy corresponds to a smaller epsilon and to a more more control of the variability in $\pitil$.

\begin{table}
  \begin{center}
    \caption{Empirical acceptance rates (\%) at the optimal parameter settings for $\AAPS$, HMC, blurred HMC (HMC-bl) and the no U-turn sampler (NUTS)\label{tab.numAcc}. $^*\xi$ for odd-numbered components of $\pi_{MR}$. $^{**}$The acceptance rate for HMC on $\pi_G^{RN}$ was extremely sensitive to $T$. An empirical estimate of the limiting acceptance rate as $\epsilon \downarrow 0$ is also supplied.}
\begin{tabular}{l|rr|ccccc}
  Target type& $d$ & $\xi$       	& $\AAPS$ &$\AAPS$-limit& HMC & HMC-bl & NUTS \\
  \hline
$\pi_G^{SD}$&40&20       	& 86.6& 83.2	&90.7 	& 78.1 	& 99.1    \\
$\pi_G^{VAR}$&40&20      	& 84.9& 80.1	&91.6 	& 88.4 	& 99.8   \\
$\pi_G^{H}$&40&20        	& 83.7& 85.4	&63.0 	& 71.7 	& 82.8    \\
$\pi_G^{invSD}$&40&20    	& 83.8& 88.3	&70.1 	& 72.4 	& 91.3    \\
\hline
$\pi_{SG}^{VAR}$&40&20            & 84.6& 84.0 &93.2 	& 90.6 	& 94.3\\
$\pi_L^{VAR}$&40&20      	& 74.7& 76.9	&73.9 	& 77.1 	& 96.3    \\
$\pi_G^{VAR}$&100&20     	& 84.5& 80.8	&78.6 	& 75.7 	& 96.4    \\
$\pi_G^{VAR}$&40&40      	& 81.6& 78.6	&97.5 	& 77.1 	& 96.8   \\
\hline
$\pi_{MR}$&20&$^*$10		& 82.8& 84.6	&78.9 	& 84.8	& 89.9   \\
$\pi_{MR}$&40&$^*$10		& 84.7& 87.4	&80.7 	& 80.0 	& 89.7 \\
$\pi_{MR}$&100&$^*$10		& 83.6& 87.1	&73.8 	& 74.3	& 98.3   \\
$\pi_{MR}$&400&$^*$10		& 83.2& 85.2	& 62.4 	& 71.4 	& 97.4 \\
\hline
$\pi_G^{RN}$&30&110		& 83.3 & 87.5	&$^{**}$62.0 	& 79.7	& 99.4    \\
\end{tabular}
\end{center}
\end{table}

Table \ref{tab.ToyTargetRelEffRec_REV} in Appendix \ref{sec.useTuningAdvice} shows the relative efficiency of AAPS when tuned according to this advice, and the advice to following (Appendix \ref{sec.tuning.K}) on tuning $K$, compared to the most efficient algorithm found via a grid search. For additional context this is compared with the efficiency of HMC, HMC-bl and NUTS when $\epsilon$ is chosen according to standard acceptance-rate criteria, and $T$ (HMC) or $T_*$ (HMC-bl) is chosen by a grid search. AAPS remains competitive with all of the other algorithms.

\section{Tuning $K$}
\label{sec.tuning.K}

We now describe a diagnostic that we have found useful for choosing the $K$ parameter, the number of segments to use over and above segment $0$. 

Each iteration, given a current position $\xcurr$, $\AAPS$ samples a value $c$ uniformly from $\{0,\dots,K\}$, AAPS proposes a point $x'$ from the set of points in the $K+1$ segments numbered $-c,\dots,K-c$, with a probability proportional to our chosen weight function, $w(\xcurr,x')\propto \pitil(x')||x'-\xcurr||^2$. It will be helpful in the sequel to consider the case where the segment number $j$ was, instead, sampled uniformly from $\{-c,\dots,K-c\}$. This would be equivalent to sampling $c$ and $e$ independently and uniformly from $\{0,\dots,K\}$ and setting $j=e-c$. Thus, the marginal distribution from which $j$ arose would be:
\[
\Prob{J=j}=
\frac{K+1-|j|}{(K+1)^2}\mbox{ if }j=-K,\dots,0,\dots,K,
~~~\mbox{and}~~~
\Prob{J=j}=0\mbox{ otherwise.}
\]

To decide upon a suitable $K$, we first perform a relatively short run of $\AAPS$ with a large $K$, $K_*$. Each iteration, we note the segment number, $j\in \{-a,\dots, K_*-a\}$ from which the proposal arose. By the symmetry of the sampling of both $a$ and the momentum, $p$, the distribution of $j$ is symmetric about $0$, so we track $k=|j|$ and keep a running total of the number of times, $n_{K_*}(k)$, that there has been a proposal from a segment with an index whose absolute value is $k$ for each $k$ from $0$ to $K_*$.

If the weights had been irrelevant and all segments contained the same number of points then
\[
p_K(k):=\Prob{|j|=k|\mbox{proposed segment sampled uniformly}, K}
=
\left\{
\begin{array}{ll}
  \frac{1}{K+1}&\mbox{if }k=0\\
  2\frac{K+1-k}{(K+1)^2}&\mbox{if }k=1,\dots,K_*,\\
  0&\mbox{otherwise}.
  \end{array}
\right.
\]
The method for placing the set of segments over the current segment inherently causes lower values of $k$ to appear more often, whatever the weights might suggest. To account for this, we calculate
\[
m_{K_*}(k)=\frac{n_{K_*}(k)}{p_{K_*}(k)}
\]
and choose the optimal number of additional apogees, $\Khat$ as
\[
\Khat=\arg \max_{k=0,\dots K_*}m_{K_*}(k).
\]
Below we give a further heuristic for why this tuning mechanism is reasonable and several plots from empirical studies which demonstrate it working in practice. In practice, during the tuning run we monitor $\mbar_{K_*}(k):=100 (K_*+1)m_{K_*}(k)/\sum_{i=0}^{K_*} m_{K_*}(i)$ since it stabilises as the number of iterations increases.

\subsection{Heuristic explanation}
For a particular segment, $j$ segments from the current point, let $s_*(j;x,p)$ represent the average (over the segment) squared distance in position between the current point, $(x,p)$, and points in the segment. To explain, heuristically, how the diagnostic works we make three simplifying assumptions:

\begin{enumerate}
\item $s_*(j;x,p)= c(x,p) s(j)$, for some functions $c$ and $s$ with $s(0)=0$.
\item The number of points in segment $j$ does not depend on $j$; \emph{i.e.},  $|\cS_j(X,P)|=N(X,P)$ for some integer-valued function $N$.
  \item Acceptance probabilities are generally large enough that variation in these is a secondary effect.
\end{enumerate}
The first assumption seems reasonable as, at least for low $j$, $s_*(j;x,p)\propto j^2$, approximately, although, strictly $s_*(0;x,p)>0$ unless there is a single point in the initial segment. The second assumption is strictly incorrect, but in the limit as $\epsilon \downarrow 0$, $\Expect{|\cS_j(X,P)|}$ does not depend on $j$ since, at stationarity, all points in the Hamiltonian path have the same density as the initial point (see the proof of Theorem \ref{thrm.GP}). Thus, the second assumption is reasonable provided values do not vary too much from $N(X,P)$. The third assumption is certainly correct in the limit as the step size, $\epsilon\downarrow 0$, but is reasonable empirically more generally.

Since proposals are made in proportion to  squared jumping distance, 
\[
m_{K_*}(j)\approx \Expects{\pitil}{s_*(j;X,P)N(X,P)}=c_1 s(j),
~~~\mbox{where}~c_1:=\Expects{\pitil}{c(X,P)N(X,P)}.
\]

We now assume that the tuning parameter has been set to $K$.
The absolute segment number $k$ is proposed with a probability proportional to $p_{K}(k) N(X,P) c(X,P) s(k)$. The mean squared jumping distance resulting from a tuning parameter $K$ is, therefore
\[
MSJD_K:=\frac{\sum_{j=0}^K p_K(j) N(X,P) c(X,P)^2 s(j)^2}{\sum_{j=0}^K p_K(j) N(X,P)c(X,P) s(j)}
=
c(X,P)\frac{\sum_{j=0}^K p_K(j) s(j)^2}{\sum_{j=0}^K p_K(j) s(j)}.
\]
Denoting $\Expect{c(X,P)}$ by $c_2$, the expectation over all initial values is
\[
\ESJD_K
:=
c_2\frac{\sum_{j=0}^K p_K(j) s(j)^2}{\sum_{j=0}^K p_K(j) s(j)}
=
c_2\frac{\sum_{j=1}^K p_K(j) s(j)^2}{\sum_{j=1}^K p_K(j) s(j)}
=
c_2\frac{\sum_{j=1}^K (K+1-j) s(j)^2}{\sum_{j=1}^K (K+1-j) s(j)}
\]
since $s(0)=0$. 
We can simplify this to 
\begin{equation}
\ESJD_K=\Expect{s(J)},~~~\mbox{where}~\Prob{J=j}\propto r_j:=(K+1-j)s(j),~j=1,\dots,K.
\end{equation}

$\ESJD$ does not take into account the computational effort, which is proportional to the number of segments, $K+1$. Hence, we define the efficiency as
\[
\Eff_K:=\frac{1}{K+1}\Expect{S(J)}.
\]

Intuitively, for small $j$, $s(j)\propto j^2$, approximately, which motivates the assumptions in the following. 
\begin{proposition}
  If $s(0)=s'(0)=0$ and $K$ is small enough that $s(j)$ is convex on $\{0,\dots,K\}$, then
\[
\Eff_K\ge \Effbar_K:=\frac{1}{K+1}s\left(\frac{K+1}{2}\right),
\]
and $\Effbar_j$ is non-decreasing on $j\in\{0,\dots,K\}$.
\end{proposition}

\begin{proof}
  Jensen's inequality gives $\ESJD_K\ge s(\Expect{J})$. Further, as $s$ is convex, for any $b\ge a$, the slope of the chord from $0$ to $b$ is at least as large as that of the chord from $0$ to $a$:
  \begin{equation}
    \label{eq.chords}
  \frac{s(b)-s(0)}{b-0}\ge \frac{s(a)-s(0)}{a-0}
  \implies
  \frac{s(b)}{b} \ge \frac{s(a)}{a},
  \end{equation}
  since $s(0)=0$. 
Thus, for $j\le (K+1)/2$, setting $b=K+1-j$ and $a=j$, we have
\[
r_{K+1-j}\ge r_j,
\]
and so $\Expect{J}\ge (K+1)/2$. Since $s'(j)\ge 0$ on $\{0,\dots,K\}$, we have
\[
\ESJD_K\ge s\left(\frac{K+1}{2}\right),
\]
proving the first part.
Reusing \eqref{eq.chords} with $b=(K+1)/2$ and $a=K/2$ gives, as required,
\[
\Effbar_K\ge \Effbar_{K-1}.
\]
\end{proof}

Indeed, straightforward algebra shows that if $s(j)=\lambda j$ then $\Eff_K=\lambda/2$; the inequality is tight for convex functions. However, in practice, $s(j)$ is \emph{strictly} convex initially, which suggests the maximum efficiency may occur after the first point when $s$ is no longer convex. Plots from various starting points of the mean squared Euclidean distance from the start to points in segment $j$ show behaviour approximately similar to
$s(j)\propto 1-\cos(\pi j/b)$, for some $b$, although there is, typically, less oscillation between peaks and troughs once the first peak has been passed. For small $j$, this formulation gives, approximately, $s(j)\propto j^2$ which fits with intuition.

Figure \ref{fig.IdealisedEff} plots $s(j)$ against $j$ when $b=15$, and the resulting $\Eff_K$ against $K$. The efficiency is maximised at a value close to $\inf_{j \ge 0}\arg \max s(j)$ and the relative difference between the efficiencies at the two points is small (here less that $0.5\%$). The $1/(K+1)$ penalty term means that damping of the oscillations of $s(j)$ after the first peak will make no difference to the point of maximum efficiency, only to the tail of the efficiency curve.

\begin{figure}
\begin{center}
    \includegraphics[scale=0.5]{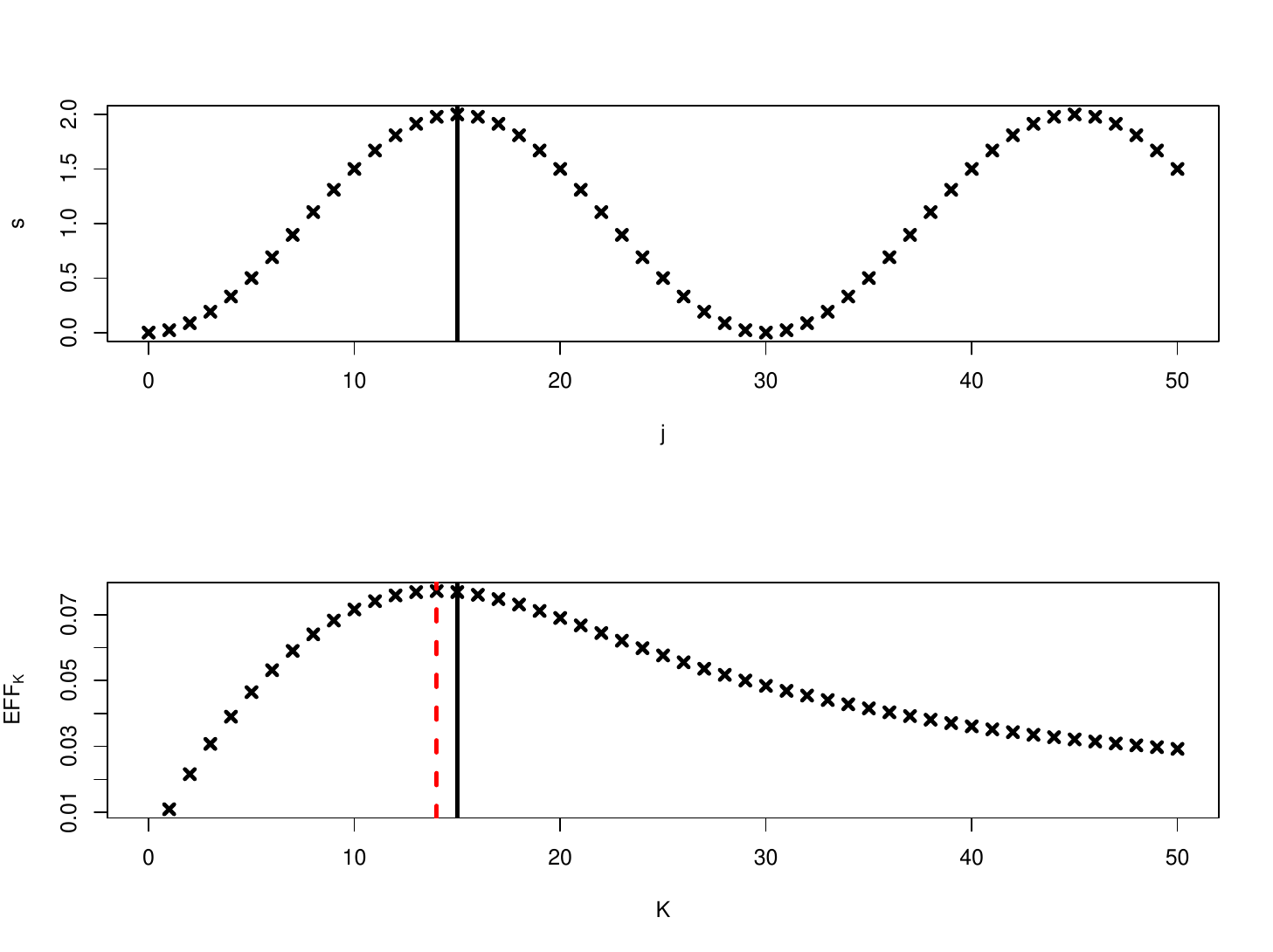}
    \end{center}
\caption{Top panel: plot of $s(j)=1-\cos(\pi j/b)$ against $j$. Bottom panel: plot of the resulting $\Eff_K$ against $K$. The vertical black line in each plot shows the value of $b$ (here, $15$), whilst the dashed vertical red line shows the $K$ value at which $\Eff_K$ is maximised.
\label{fig.IdealisedEff}}
\end{figure}

\subsection{Empirical verification}

\begin{figure}
\begin{center}
    \includegraphics[scale=0.6]{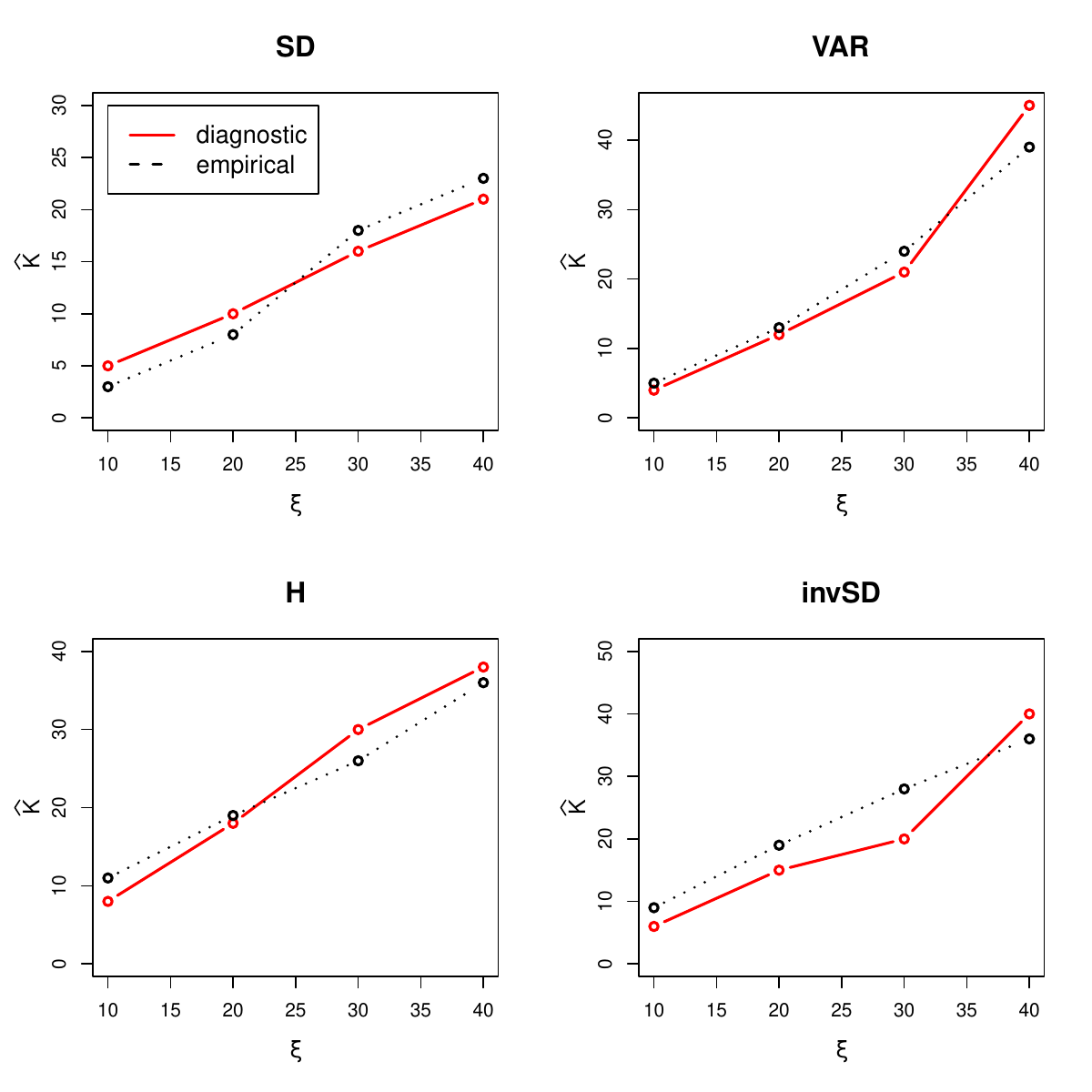}
    \end{center}
\caption{The optimal $K$ found by testing over a fine grid of possible values compared with predictions based  on the diagnostic using a $K^*=60$. All targets had dimension $d=40$ and were a product of skew Gaussians. Each run of $\AAPS$ was for at least  $10^5$ iterations.
\label{fig.empDiagnostic}}
\end{figure}

Figure \ref{fig.empDiagnostic} compares the optimal $K$ found by using a fine grid of $(K,\epsilon)$ values, with the value suggested by our diagnostic. All plots are for the skew-Gaussian product targets described in Section \ref{sec.ToyTargets}, but using a wider variety of eccentricity parameters, $\xi$, and all plots show good agreement between the diagnostic and the empirical estimate of the truth. Similar agreement was found for other targets from Section \ref{sec.ToyTargets} that we investigated.

The linear increase of $K$ with $\xi$ can be explained by the following heuristic for which we would like to thank one of the reviewers: rescaling if necessary, in Corollary \ref{cor.EN} consider $\Expect{\nu}=1$ so that the largest scale parameters are $\propto \xi^{1/2}$ and the smallest are $\propto \xi^{-1/2}$. Notice that the expected number of apogees per unit time interval is proportional to $1/\sqrt{\Expect{\nu^2}}$, the square root of the harmonic mean of the squared scale parameters. This is typically dominated by, and approximately proportional to the mean of the smallest scale parameters. By contrast, the time  required for a large amount of total movement is dominated by that needed for movement in the largest components. The quotient of these two quantities is $\xi$.

\section{Bimodal Targets}
\label{sec.bimodal}

The AAPS algorithm with $K=0$ is reducible on a multimodal, one-dimensional target as it cannot travel between the modes; however Figure \ref{fig.bimodalApogees} illustrates that even in $d=2$, AAPS with $K=0$ is not reducible. The figure depicts the path under Hamiltonian dynamics, marking on the apogees and clearly showing paths which cross from one potential well into another between apogees. The approximate dynamics using the leapfrog step can allow movement between the two wells in between two apogees by this mechanism, but also, when $\epsilon$ is large, the dynamics sometimes miss an apogee entirely, which again allows movement between the models.

\begin{figure}
\begin{center}
  \includegraphics[scale=.6,angle=0]{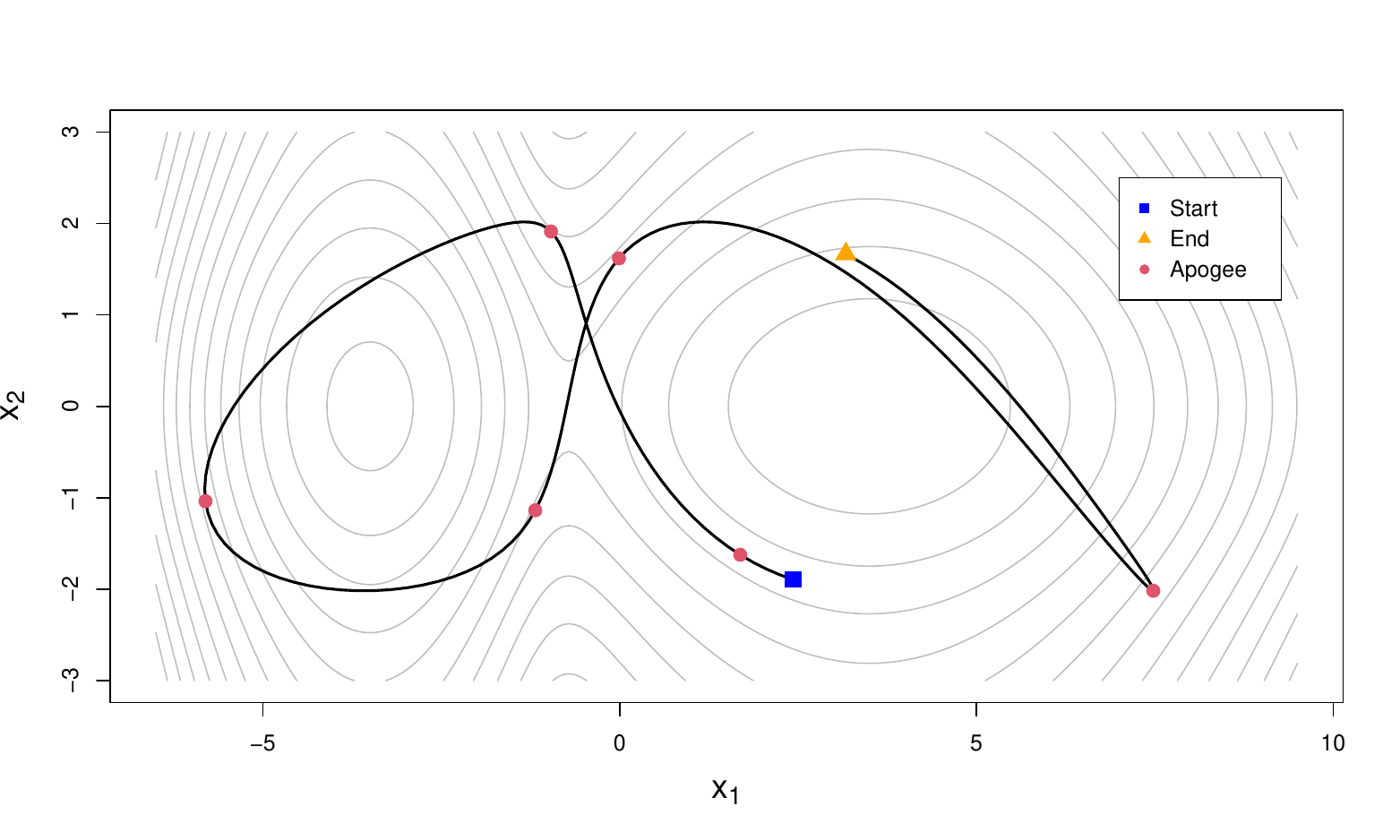}
\end{center}
\caption{The true path under Hamiltonian dynamics on a potential from a target of the form \eqref{eq.bimodal} with $d=2$ and $a=3.5$, but the variance of the second mixture component is $4I_2$. The particle starts at the blue square and ends at the yellow triangle with apogees marked as red dots.
\label{fig.bimodalApogees}}
\end{figure}

To compare algorithm efficiencies on bimodal targets we ran $\AAPS$, blurred HMC and the no U-turn sampler on targets of the form \cite[][rotated]{PoHo2020}:
\begin{equation}
\label{eq.bimodal}
  X\sim
\frac{1}{2}\mathsf{N}\left((-a,0,\dots,0)^\top,I_{40}\right)
+
\frac{1}{2}\mathsf{N}\left((a,0,\dots,0)^\top,100 I_{40}\right),
\end{equation}
with three different values for $a$ from $7$ (modes barely separated) to $15$ (substantial separation). 
Table \ref{tab.bimodal} shows the efficiencies of the optimised algorithms in each case. 

   \begin{table}
    \caption{$\mathsf{ESS}(x_1)/n_{\mathsf{leap}}\times 10^5$ for optimally tuned versions of AAPS, blurred HMC and the no U-turn sampler (NUTS) applied targets of the form \eqref{eq.bimodal} with three different values for $a$.
    \label{tab.bimodal}}
   \begin{center}
    \begin{tabular}{l|rrr}
      a&HMC-bl & AAPS &NUTS\\
      \hline
     7 & 347.4 & 231.4 & 209.4 \\
     10   & 158.2 & 119.6 & 82.3  \\
     15  & 37.2  & 20.9  & 14.5  
    \end{tabular}
  \end{center}
   \end{table}
   
  Blurred HMC is more efficient than AAPS which is more efficient than the no U-turn sampler; however, AAPS is always at least half as efficient as HMC.

\section{Weighting Scheme 6}
\label{sec.WSsix}

Denote the points in $\cS_{a:b}(\zcurr)$ from furthest back in time to furthest forward in time as $z_B,\dots,z_0,\dots,z_F$, and for $l\in\{B,\dots,F\}$ let $T_l:=\sum_{i=B}^l \pitil(z_i)$.

\begin{enumerate}
\item Let $\cH_0:=\{z_l:T_l\le T_F/2\}$ and $\cH_1:=\{z_l:T_{l-1}\ge T_F/2\}$. For  $z\in \cH_0\cup \cH_1$, the corresponding weight is $w(z)=\pitil(z)$.
\item In the special (null) event that $\exists h: T_h=T_F/2$, $\cH_0\cup \cH_1=\cS_{a:b}(\zcurr)$ go to Step 6; otherwise define the boundary point as $z_h$, where $T_{h-1}<T_F/2< T_h$.
\item Split the point $z_h$ into two identical points, $z_h^B=z_h$, which has a weight of $w(z_h^B)=T_F/2-T_{h-1}$ and $z_h^F=z_h$ which has a weight of $w(z_h^F)=T_h-T_F/2$. Thus $w(z_h^B)+w(z_h^F)=\pitil(z_h)$.
  \item Set $\cH_0\gets \cH_0 \cup \{z_h^B\}$ and $\cH_1=\cH_1\cup\{z_h^F\}$.
\item If $\zcurr=z_h$, with a probability of $w(z_h^B)/\pitil(z_h)$ decide that $\zcurr \in \cH_0$; else decide that $\zcurr \in \cH_1$.
\item For $i\in\{0,1\}$ if $\zcurr \in \cH_i$ then $\cH(\zcurr)=\cH_{1-i}$.
  \item Sample a point from $\cH(\zcurr)$ with a probability proportional to the weight.
\end{enumerate}

\section{Further numerical examples}
\label{sec.useTuningAdvice}
Table \ref{tab.ToyTargetRelEffRec_REV} shows the efficiencies of the different algorithms when tuned according to recommended guidelines.

For HMC and NUTS, this involved picking the $\varepsilon$ which resulted in a desired target acceptance rate; $65.1\%$ for blurred and unblurred HMC, and $80\%$ for NUTS. Once appropriate $\varepsilon$ values were found, the integration time $T$ for both variants of HMC was optimised on a fine grid of values. In $\AAPS$, we first identified a stable $\varepsilon$ and ran the algorithm to identify a good choice for the number of segments using the diagnostic described in Appendix \ref{sec.tuning.K}. We then used the advice in Section \ref{sec.TuningEpsilon} to tune the $\varepsilon$ based on the discrepancy from the limiting acceptance rate; we  set the threshold of the absolute difference at $3\%$. At this optimal $\epsilon$, the number of segments was further refined via grid search in the neighbourhood of the value provided by the diagnostic.

\begin{table}
\centering
\caption{\label{tab.ToyTargetRelEffRec_REV}Relative efficiency at recommended algorithm settings compared with AAPS of HMC, blurred HMC (HMC-bl) and the no U-turn sampler (NUTS); raw efficiencies are the quotient of that from \eqref{eqn.gen.eff} and the efficiency of the optimally tuned (via a grid search) AAPS algorithm. $^{(1)}$ The final acceptance rates were higher than recommended as further increasing $\varepsilon$ resulted in unstable leapfrog integration.  $^{(2)}\xi$ for odd-numbered components of $\pi_{MR}$.}
\begin{tabular}{l|rr|c|cccr}
  Target type&$d$&$\xi$&$\AAPS_{opt}$&$\AAPS_{3\%}$&HMC&HMC-bl&NUTS\\
  \hline
  $\pi_G^{SD}$&40&20&1.000&0.904  & 0.617 & 0.539 & $^{(1)}$0.796\\
  $\pi_G^{VAR}$&40&20&1.000& 0.809 & 1.016 & 0.820 & $^{(1)}$0.844\\
  $\pi_G^H$&40&20&1.000  & 1.000 & 0.162 & 0.601 & 0.387\\
  $\pi_G^{invSD}$&40&20&1.000&  0.984  & 0.141 & 0.457 & 0.417\\
  \hline
  $\pi_{SG}^{VAR}$&40&20&1.000& 0.804 &  0.790 & 0.866 & 0.460\\
  $\pi_{L}^{VAR}$&40&20&1.000& 0.961  & 1.102 & 1.329 & 0.427 \\
   $\pi_G^{VAR}$&100&20&1.000& 0.730  & 0.649 & 1.021 & 0.901\\
  $\pi_G^{VAR}$&40&40&1.000 & 1.000  & 0.924 & 0.778 & 1.073\\
  \hline
  $\pi_{MR}$&20&$^{(2)}$10&1.000& 1.000  & 1.569 & 1.386 & 0.603\\
    $\pi_{MR}$&40&$^{(2)}$10&1.000  & 0.992 & 0.924 & 1.032 & 0.629\\
  $\pi_{MR}$&100&$^{(2)}$10&1.000&  0.962 & 0.755 & 0.893 & 0.404\\
  $\pi_{MR}$&400&$^{(2)}$10&1.000 & 1.000  & 0.625 & 0.803 & 0.520\\
  \hline
  $\pi_{G}^{RN}$&30&110&1.000  & $^{(1)}$0.942 & $<0.001$ & 1.072 & 0.341
\end{tabular}

\end{table}

\end{document}